\documentclass[10pt,prx,twocolumn,superscriptaddress,aps,longbibliography]{revtex4-2}

\usepackage[T1]{fontenc}
\usepackage[british]{babel}
\usepackage{amsmath}
\usepackage{amssymb}
\usepackage{amsthm}
\usepackage{graphicx}
\usepackage{bm}
\usepackage{booktabs}
\usepackage{multirow}
\usepackage{tikz}
\usetikzlibrary{arrows.meta, positioning, calc}
\usepackage[colorlinks=true, allcolors=blue]{hyperref}

\newcommand{\FAB}{F_{\!AB}}
\newcommand{\FAE}{F_{\!AE}}
\newcommand{\Holevo}{\chi}
\providecommand{\ket}[1]{\ensuremath{\left|#1\right\rangle}}
\providecommand{\bra}[1]{\ensuremath{\left\langle#1\right|}}
\DeclareMathOperator{\Tr}{Tr}
\newtheorem{proposition}{Proposition}

\begin{document}

\title{Learnt Attacks on Quantum Key Distribution under\texorpdfstring{\\}{ }%
Channel Noise and Device Drift}

\author{Marcel Mordarski\textsuperscript{\dag}}
\email{marcel.mordarski25@imperial.ac.uk}
\affiliation{Department of Computing, Imperial College London, London SW7 2AZ, United Kingdom}

\author{Benjamin Gras}
\thanks{These authors contributed equally to this work.}
\author{Abdelrahman Shehata\textsuperscript{\dag}}
\author{Daniel Budina\textsuperscript{\dag}}
\affiliation{Department of Electrical and Electronic Engineering, Imperial College London, London SW7 2AZ, United Kingdom}

\author{Roberto Bondesan}
\affiliation{Department of Computing, Imperial College London, London SW7 2AZ, United Kingdom}

\begin{abstract}
Quantum key distribution (QKD) links are provisioned from security analyses of stationary channels, whereas the devices that determine the channel drift between recalibrations. Whether an eavesdropper who cannot alter the channel's own noise gains by following that drift has not been quantified. Adaptive eavesdropping is posed here as a constrained Markov decision process in which the attacker selects one circuit per round while the noise level follows an Ornstein--Uhlenbeck process and the abort condition is a budget over each block of rounds. The value of adaptation is bounded by the best fixed circuit and a dynamic-programming upper bound. The actions are learnt attacks. Whereas Decker \textit{et al.} trained a parametrised circuit on a fixed gate template against a fixed channel, here the gate structure and rotation angles are searched jointly. This yields circuits compact enough to form a discrete action set, extending the construction to noise models lacking a known template, including the amplitude damping channel. On device-independent E91 under bilateral depolarising noise, a reinforcement-learning attacker raises her Holevo information from $0.135$ for the best fixed circuit to $0.348$ at zero detection, $98\%$ of the upper bound. On BB84 under a drifting bit-flip channel, she exceeds a conservative noise-indexed rule by $0.024$ in fidelity, reaching $99\%$ of the upper bound. Under stationary noise, the attacker's gain from basis asymmetry changes sign between an averaged and a per-basis error-rate constraint. The search, started from random gate sequences, recovers the analytical cloners and the collective-attack key rate, and meets the lower bound of the Winick--L\"utkenhaus--Coles objective from above.
\end{abstract}

\maketitle

\section{Introduction}\label{sec:intro}

\subsection{Problem construction}\label{intro_prob_const}

Quantum key distribution (QKD)~\cite{bennett1984bb84, ekert1991e91} now secures metropolitan links between bank data centres and government sites, and commercial hardware is sold against published security certificates~\cite{Scarani2009}. In every such deployment, the legitimate parties, Alice and Bob, estimate the disturbance of the transmitted states on a sample of their rounds, and a security proof converts that estimate into a bound on the information an eavesdropper, Eve, can hold about the key~\cite{Metger2022, Tomamichel2015}. The user sets an abort threshold on the estimated disturbance and a finite-key buffer so that the bound remains satisfied, and recalibrates the source and detectors at an interval intended to keep the devices within their characterised specification~\cite{Tomamichel2015, Trushechkin2022}.

The analyses underlying these settings treat the channel as stationary between calibrations. A device-dependent proof such as that of BB84 requires the source and detectors to remain within calibrated bounds (e.g., efficiency matching of the detectors and the fidelity of state preparation)~\cite{Trushechkin2022, Wang2023}, whereas the statistics on which a finite-key analysis rests are gathered over a window during which the equipment drifts~\cite{George2022}. Deployments accommodate the drift by recalibrating on a fixed schedule and carrying a margin wide enough to cover the interval, and the width of that margin is set from device characterisation, not from an analysis of what an attacker could gain from the drift~\cite{George2022, Trushechkin2022}.

The numerical methods that certify key rates answer a static question. Given one observed behaviour, the relative-entropy method of Winick, L\"utkenhaus and Coles~\cite{winick2018reliable} and the semidefinite hierarchy of Navascu\'es, Pironio and Ac\'in~\cite{navascues2008npa, brown2021npa} bound how much Eve can know. A drifting channel poses a sequential question in addition, since an eavesdropper who reads the public record may change her strategy from round to round, must commit to each round's strategy before that round's statistic is announced, and seeks the largest expected information over the trajectories of the channel subject to a constraint in every round. A relaxation evaluated at one frozen behaviour describes no such policy, and the value that adaptation confers on an attacker confined to the physical noise process has, to the authors' knowledge, not been quantified. The margin carried across a recalibration interval therefore rests on an untested assumption about the attacker, which leaves the user unable to tell whether a link is over-provisioned, at a loss of secret-key rate, or exposed to a leak that a stationary trusted-noise analysis does not count.

The question concerns the operating margin and leaves the validity of security proofs intact. A proof that attributes the entire observed disturbance to Eve bounds her information from the statistics actually recorded, and entropy-accumulation arguments extend such bounds to strategies that adapt from round to round~\cite{Metger2022}, so that under this convention no leak goes uncounted, whether or not the channel drifts. What drift alters is the distance between that bound and the information available to an attacker who cannot alter the channel's own noise, which is the attacker assumed whenever part of the observed noise is characterised and treated as trusted~\cite{lin2020trusted}. That distance determines how conservative the provisioned margin is.

While sectors such as banking, government, and civil telecommunications increasingly seek to deploy QKD, these ambitions are stalled by the prohibitive costs of maintaining continuous secure links. Without effective optimisation approaches to rigorously bound an adaptive eavesdropper, network operators must rely on worst-case static assumptions. This forces them to provision excessively conservative margins, resulting in frequent network downtime for unnecessary recalibrations, artificially depressed secret-key rates, and severe underutilisation of expensive hardware. Resolving this ambiguity is therefore a practical necessity: translating dynamic threats into measurable, optimisable quantities allows operators to safely extend recalibration intervals and maximise throughput, thereby clearing one of the primary economic barriers to the commercial adoption of quantum-secure infrastructure.

\subsection{Prior work}\label{intro_prior_work}

Decker \textit{et al.}~\cite{decker2024qkd} observed that an attack on BB84 is itself a parametrised quantum circuit and can therefore be learnt. In their formulation, the eavesdropper applies a unitary $U(\Theta)$ to the travelling qubit together with an ancilla she retains, and a second unitary $V(\Lambda)$ after the public basis announcement, both drawn from a hardware-efficient template of rotations and \textsc{cnot}s. Gradient descent on
\begin{equation}
\mathcal{L}(\Theta, \Lambda) = \alpha\!\left(\FAB - f\right)^2 - \FAE
\label{eq:decker}
\end{equation}
pins Bob's fidelity $\FAB$ to a target $f$ through the first term while maximising the eavesdropper's fidelity $\FAE$ through the second. Three of their results serve as the starting point here. Firstly, the learner recovers the analytical cloner on an ideal channel. Secondly, under a bit-flip channel, an \emph{imbalanced} cloner carrying two free parameters exceeds it by giving up fidelity in the damaged basis to preserve fidelity in the untouched one. Thirdly, when the parity information leaked by Cascade reconciliation~\cite{brassard1993secret} is treated as available to the attacker, a joint measurement raises her recovery probability from $0.810$ to the Helstrom value $0.894$~\cite{helstrom1969quantum} at the $11\%$ error bound. Section~\ref{sm:decker} of the Supplemental Material reviews the construction in the notation used here.

For a stationary Pauli channel, whose noise applies each Pauli operator with a fixed probability, the eavesdropper's side is well charted. The optimal individual attack on the noiseless channel is the phase-covariant cloner~\cite{fuchs1997optimal, bruss2000phasecov}, Cerf's Pauli cloning machines supply the asymmetric family natural to Pauli noise~\cite{cerf2000cloning}, the underlying cloners admit explicit circuit realisations~\cite{griffiths1997circuit, niu1998copying}, and a result from the group of Ref.~\cite{decker2024qkd} gives explicit constructions on a fixed Niu--Griffiths circuit for asymmetric cloners tailored to Pauli noise, the biased cloner among them~\cite{kerstan2026pauli}. On the key-rate side, the Shor--Preskill rate $1 - h(Q_Z) - h(Q_X)$, in which $Q_Z$ and $Q_X$ are the error rates observed in the key basis and in the complementary basis and $h$ is the binary entropy, governs independent bit and phase errors~\cite{shor2000simple}, and for a family of imperfect BB84 sources an explicit attack attains the corresponding collective-attack bound with equality~\cite{woodhead2014tight}. Upper bounds on device-dependent rates follow from entropies of two-qubit states~\cite{kraus2005lower, renner2005information}, including for imperfect detectors~\cite{moroder2006upper}. In the device-independent setting, the Ac\'in--Brunner--Gisin--Massar--Pironio--Scarani envelope~\cite{acin2007diqkd} is saturated by an explicit attack~\cite{pironio2009diqkd, primaatmaja2023review}, and convex-combination attacks place constructive upper bounds on device-independent rates at arbitrary visibility and detection efficiency~\cite{lukanowski2023convex, farkas2021bell}. The search introduced below recovers these results where they apply, and the recovery serves here as validation (Sec.~\ref{sm:validation}). Learning has also been applied to eavesdropping directly: variational quantum cloning casts the attack as a trained circuit~\cite{coyle2020varqcloning}, the formulation of Ref.~\cite{decker2024qkd} is the point of departure here, and a recent study trains a tabular agent to modulate a fixed intercept-resend attack against the broadcast error rate~\cite{bagunu2026adaptive}. The latter shares the present premise that an eavesdropper reading the public record is a sequential learner, but its actions are one physical attack at a chosen rate rather than a library of constructed circuits, and its channel is stationary rather than drifting (as elaborated in Sec.~\ref{sm:concurrent}).

The present work uses the learnt formulation of Ref.~\cite{decker2024qkd} as a tool for a different question, and two of its restrictions must be lifted for that purpose. The circuit template is fixed in advance, so that only angles are explored within a topology chosen for tractability, and the channel is held at one stationary noise level, so that no statement about a drifting device follows. For Pauli noise, the first restriction is benign, because a fixed template already reaches the optimum~\cite{kerstan2026pauli}, and a searched topology is justified where no template is known (e.g., amplitude damping) as in Sec.~\ref{sec:nonpauli}. The second restriction excludes the drifting channel of Sec.~\ref{intro_prob_const}, for which neither the learnt nor the analytical treatment returns a strategy, since each evaluates one fixed observed behaviour.

\subsection{Methodology}\label{sec:methodology}

The learnt formulation is extended under two noise conventions. The channel sweeps and the drift experiments treat the noise as \emph{trusted}, meaning that the eavesdropper is assigned a physical noise process she cannot alter, together with an ancilla of bounded size, and the quantity of interest is what she extracts under those constraints. The key-rate comparisons of Secs.~\ref{sm:validation} and~\ref{sm:keyrates} adopt the standard \emph{untrusted} convention, which attributes the entire observed disturbance to the eavesdropper~\cite{Scarani2009}. The first convention measures what a physically restricted attacker achieves on a degrading link, and hence how conservative a provisioned margin is, whereas the second measures what a formal security proof must bound. Results under the trusted convention are therefore statements about the margin and never about the soundness of a proof. Every value reported under either convention is a constructive lower bound on her power, since an explicit circuit exhibits it, and through the Devetak--Winter relation~\cite{devetak2005distillation} each collective attack yields an upper bound on the key rate.

The first departure from the learnt formulation abandons the fixed template. A template selected for tractability under one noise model presupposes the form of the attack, since the extremal circuit under another model is not obliged to share its gate structure. The topology is therefore treated as an unknown on the same footing as the angles, giving
\begin{equation}
\min_{T,\,\theta}\; f(T, \theta),
\label{eq:joint}
\end{equation}
where $T$ is a gate sequence over $\{R_X, R_Y, R_Z, \textsc{cnot}\}$ subject to no symmetry, depth or structural constraint, $\theta$ collects its angles, and $f$ is the negative of the eavesdropper's information measure. The publicly monitored statistic enters as a strict feasibility constraint, so that a circuit exceeding the abort threshold is excluded outright and can never outrank a feasible one. A soft penalty would allow a larger leak to be reported at a detection probability the defender does not tolerate, and such a leak could not be sustained against a user who aborts.

Equation~\eqref{eq:joint} is solved by nesting two loops. An evolutionary outer loop proposes gate sequences through moves that add, remove or replace one gate~\cite{real2019regevo, chivilikhin2020mogvqe}, and for each candidate an inner loop tunes the angles by a gradient-free optimiser at a fixed evaluation budget, covariance-matrix adaptation~\cite{Hansen2001, hansen2016cmaes} for the BB84 sweeps and simultaneous-perturbation stochastic approximation~\cite{spall1992multivariate} for the E91 library. The two are interchangeable at fixed budget (Sec.~\ref{sm:framework}), so the choice affects neither the returned circuits nor the reported values; the gate alphabet and the fitness are what change between the protocols, and Sec.~\ref{sm:framework} gives the architecture, hyperparameters and cost. The procedure returns a circuit that achieves a stated leak at a stated disturbance, which is a weaker guarantee than a proof of optimality but is available where no proof is.

The second departure replaces the stationary channel by a drifting one. The noise parameter follows an Ornstein--Uhlenbeck process, a mean-reverting Gaussian process~\cite{uhlenbeck1930theory}, and the legitimate parties abort when the monitored statistic, averaged over a block of rounds, crosses its threshold. Because the returned circuits are small, of four to nine gates, a library of them forms a discrete action set, and adaptive eavesdropping becomes a constrained Markov decision process, a sequential decision problem in which every action yields a reward and the actions of a block share one constraint (Sec.~\ref{sec:drift_formulation}). Solving it with a value-based reinforcement-learning agent~\cite{hessel2018rainbow} measures the information that adaptation yields to the attacker, and comparing the result with the best fixed circuit and with a dynamic-programming upper bound bounds that value from both sides (Proposition~\ref{prop:ordering}). This sequential formulation is the principal departure from Ref.~\cite{decker2024qkd} and from the static analyses of Sec.~\ref{intro_prior_work}, which evaluate one channel at a time and return no policy.

\subsection{Contributions}

The contributions are five. Firstly, the attainable information of an attacker confined to a stationary noisy channel is mapped across channel families at a fixed abort threshold, together with its separation from the noiseless cloning figure a user would quote (Sec.~\ref{sec:noisy_channels}). Secondly, the sign of the attacker's gain from basis asymmetry is shown to depend on the monitoring convention, positive under an error-rate constraint averaged over the two bases and negative once each basis is constrained separately (Sec.~\ref{sec:asymmetry}), and the search is applied without modification to amplitude damping, for which the Pauli-channel cloning solutions do not apply (Sec.~\ref{sec:nonpauli}). Thirdly, the bilateral structure of the channel in the entanglement-based protocol is derived, which fixes a noise level beyond which even an undisturbed link yields no key (Sec.~\ref{sec:di_channel}). Fourthly, adaptive eavesdropping under drift is posed as a constrained Markov decision process over libraries of learnt circuits and solved for BB84 and for device-independent E91, with the value of adaptation bounded between fixed play and a dynamic-programming upper bound; on E91 adaptation multiplies the attacker's information by about $2.6$ at unchanged detection (Secs.~\ref{sec:drift_formulation}--\ref{sec:drift_e91}). Fifthly, the role of tolerated detection is examined, and on both protocols the upper bound is shown to be reachable without it (Sec.~\ref{sec:thresholds}). The validation of the search against known closed forms and against a reliable numerical lower bound is collected in Sec.~\ref{sm:validation}.

The central quantity is the additional information an attacker obtains by adapting to the drift when she cannot alter the channel's own noise. The choice of a recalibration interval depends on it, and to the authors' knowledge no construction for it was previously available. The measurement requires no closed form for the channel, so the same apparatus applies to noise models and drift processes beyond those treated here.

\section{Protocols, channels and figures of merit}\label{sec:setting}

In the BB84 protocol, Alice prepares one of the four states $\ket{0}$, $\ket{1}$, $\ket{+}$, and $\ket{-}$ with equal probability. The eavesdropper applies a unitary transformation, $U_E$, to the travelling qubit together with an ancilla of $n_a$ qubits initialised in $\ket{0}^{\otimes n_a}$ (where $n_a = 2$ unless stated otherwise). Crucially, this interaction occurs before the signal passes through the noisy channel, meaning Eve interacts with the undegraded signal. Because the channel acts only after her interaction, her reduced state remains untouched. Consequently, Eve's figure of merit depends solely on her attack circuit, whereas Bob's depends on both the circuit and the channel. Bob's fidelity $\FAB$ is the mean over the four states of his probability of recovering Alice's bit when he measures in the matched basis. The eavesdropper's fidelity $\FAE$ is the mean over the two bases of the Helstrom success probability
\begin{equation}
P_{\mathrm{H}} = \tfrac{1}{2} + \tfrac{1}{4}\left\|\rho_E^{(0)} - \rho_E^{(1)}\right\|_1,
\label{eq:helstrom}
\end{equation}
which is the largest probability with which she can discriminate her two conditional memory states $\rho_E^{(0)}$ and $\rho_E^{(1)}$ at equal priors once the basis is announced~\cite{helstrom1969quantum}. The abort threshold is $\FAB \geqslant 0.87$ throughout. It corresponds to a basis-averaged error rate of $13\%$, beyond the $11\%$ at which the symmetric Shor--Preskill rate vanishes~\cite{shor2000simple}, and beyond the $12.4\%$ reached with noisy preprocessing~\cite{renner2005information}. Therefore, it is not an operational threshold for one-way post-processing under the untrusted convention, although two-way post-processing tolerates up to $18.9$--$20\%$~\cite{gottesman2003twoway, chau2002practical}. It is chosen to leave room for the channel error swept in Sec.~\ref{sec:stationary}, since under the trusted convention the attacker's reach is set by the interval between the honest fidelity and the threshold, and the search is validated at three further thresholds in Sec.~\ref{sm:validation}.

Two classes of attack appear in what follows. Because $\FAE$ is the success probability of a measurement performed on a single ancilla, results reported in $\FAE$ concern \emph{individual} attacks, for which the closed form of Refs.~\cite{fuchs1997optimal, cerf2000cloning} is optimal. The Holevo quantity $\Holevo(A{:}E)$ bounds the information obtainable by a joint measurement over many stored ancillas, so results reported in $\Holevo$ concern \emph{collective} attacks and are the ones converted into key rates (Sec.~\ref{sm:keyrates}). Since the collective class contains the individual one and the class of coherent attacks contains both, an explicitly exhibited collective attack bounds from above the rate secure against coherent attacks as well, and no de Finetti or entropy-accumulation argument is required for such upper bounds to hold.

In E91 and its device-independent form, Alice and Bob share entangled pairs and security rests on a Bell test~\cite{bell1964einstein}. The monitored statistic is the Clauser--Horne--Shimony--Holt (CHSH) value $S$~\cite{clauser1969proposed}, a combination of four correlators between Alice's and Bob's measurement outcomes that is at most $2$ for local models and at most $2\sqrt{2}$ in quantum theory~\cite{cirelson1980quantum}, and the protocol aborts when $S$ falls below a threshold $\tau_S$. The attacker's information is the Holevo quantity $\Holevo(A{:}E)$ between Alice's key bit and the eavesdropper's system~\cite{ekert1991e91, acin2007diqkd}. Table~\ref{tab:mapping} of Sec.~\ref{sec:drift_formulation} sets the two protocols side by side.

\section{Attacks under stationary channel noise}\label{sec:stationary}

Before introducing channel drift, the search evaluates stationary channels with a fixed amount of inherent error. The maximum information attainable at each noise level establishes a frontier. As the noise level changes, the adaptive attacker detailed in Sec.~\ref{sec:drift} exploits movement along this frontier.

\subsection{Attainable information on noisy channels}\label{sec:noisy_channels}

With the abort threshold fixed at $\FAB \geqslant 0.87$ and a total channel error weight $\varepsilon$ distributed between the two bases, the largest $\FAE$ attained by any feasible circuit lies well below the noiseless cloning envelope $\FAE = \tfrac{1}{2} + \sqrt{\FAB(1-\FAB)}$ of Refs.~\cite{fuchs1997optimal, cerf2000cloning} evaluated at the observed fidelity. That envelope depends only on $\FAB$ and returns $0.836$ at the threshold, whereas the attainable value falls from $0.828$ at $\varepsilon = 0.02$ to $0.786$ at $\varepsilon = 0.10$ and $0.664$ at $\varepsilon = 0.22$, so the separation grows from $0.0085$ through $0.050$ to $0.172$ as in Fig.~\ref{fig:gap}(a) and Table~\ref{tab:sweep}.

Two qualifications contextualise this separation. Firstly, the comparison assumes the trusted-noise convention (Sec.~\ref{sec:methodology}), which models the channel as a fixed process the eavesdropper cannot alter. The theoretical envelope, however, was derived assuming she has full access to a noiseless channel. Consequently, the separation measures the conservatism of the quoted figure for a physically restricted attacker, rather than implying any looseness in published bounds. Secondly, for Pauli noise, the attainable values have an analytical counterpart tied to the imbalanced and biased cloners of Refs.~\cite{decker2024qkd, kerstan2026pauli}. A disturbance $D_b$ introduced by the attacker in basis $b$ is observed as $Q_b = p_b + (1-2p_b)D_b$, where $p_b$ is the channel's inherent error rate in that basis. Composing this observation with the per-basis cloning trade-off gives
\begin{equation}
\FAE = \tfrac{1}{2} + \tfrac{1}{2}\sum_{b \in \{Z, X\}} \sqrt{D_b(1-D_b)},
\label{eq:composed}
\end{equation}
maximised over $(D_Z, D_X)$ subject to the feasibility constraint. Every entry of Table~\ref{tab:sweep} agrees with this expression to within $2\times10^{-4}$, although the search is not supplied with it. The Pauli sweeps therefore validate the search on noisy channels. The constructive content of Fig.~\ref{fig:gap}(a) lies in the magnitude of the separation and in the returned circuits, which populate the library of Sec.~\ref{sec:drift}. Furthermore, the search is essential where no such expression is known, such as for amplitude damping (Sec.~\ref{sec:nonpauli}). The underlying mechanism is the interval between the honest fidelity (which is $1 - \varepsilon/2$ for these channels) and the threshold. At $\varepsilon = 0.22$, the honest channel already sits at $\FAB = 0.89$, leaving the attacker just $0.02$ of fidelity to disturb. In contrast, the theoretical envelope attributes the entire distance from $\FAB = 1$ down to the threshold to her. This interval shifts whenever the channel's error changes, and Sec.~\ref{sec:drift} measures what an attacker obtains by tracking it.

\begin{figure}[t]
\centering
\includegraphics[width=\linewidth]{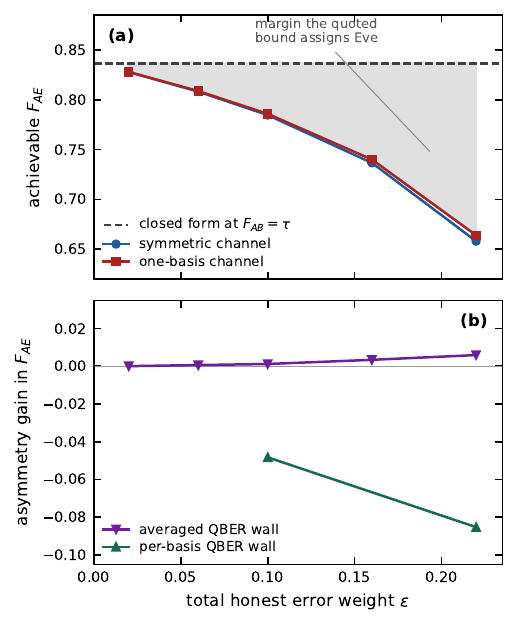}
\caption{Attainable attacks under stationary channel noise at an abort threshold of $\FAB \geqslant 0.87$. (a)~The maximum attainable $\FAE$ under the trusted-noise convention as the channel's total honest error weight, $\varepsilon$, increases. This is plotted for two scenarios: a channel concentrating all its error in a single basis, and one splitting the error evenly. Because both channels produce the same observed disturbance, the noiseless closed-form solution (dashed line) assigns them identical values. The shaded region highlights the gap between this theoretical limit and the actual attainable $\FAE$. Each data point represents the best outcome from three random seeds, with a cross-seed spread of approximately $0.001$ at $\varepsilon = 0.10$ (Sec.~\ref{sm:sweeps}). (b)~The attacker's gain from basis asymmetry under two different monitoring conventions. When constrained by the basis-averaged quantum bit error rate (QBER), the gain is small and positive. However, when constraints are applied to the per-basis rates separately (shown here under the untrusted convention where the attacker must reproduce each observed rate), the gain changes sign and increases by more than an order of magnitude. The trusted per-basis values in Table~\ref{tab:perbasis} exhibit this same sign change.}
\label{fig:gap}
\end{figure}

\subsection{Monitoring conventions and basis asymmetry}\label{sec:asymmetry}

Holding $\varepsilon$ fixed while shifting the error between the bases produces two quantum channels. A user monitoring only the basis-averaged quantum bit error rate (QBER) cannot distinguish between them, as both yield the same $\FAB$. Under this convention, the attainable leak differs slightly between the two channels. This occurs because any disturbance the attacker introduces in an already damaged basis is scaled by $1-2p_b$ in Eq.~\eqref{eq:composed} and is partially masked by the channel's intrinsic error. Consequently, her gain from this basis asymmetry rises from below $10^{-4}$ at $\varepsilon = 0.02$, to $0.0012$ at $\varepsilon = 0.10$, and reaches $0.0059$ at $\varepsilon = 0.22$.

Security proofs and most implementations instead estimate the error rate in each basis separately, and under a per-basis constraint the sign of the gain reverses, with a magnitude larger by more than an order of magnitude. Two versions of that constraint are reported in Table~\ref{tab:perbasis}. Under the trusted convention, with the key-basis rate $Q_Z$ and the complementary rate $Q_X$ each held at or below the same $13\%$, Eq.~\eqref{eq:composed} gives a gain of $-0.0049$ at $\varepsilon = 0.06$ and $-0.0214$ at $\varepsilon = 0.10$, and for $\varepsilon > 0.13$ the one-basis channel fails the per-basis test with no attacker present. Under the untrusted convention of a security proof, in which the attacker supplies the whole disturbance and must reproduce each observed rate to within $2\times10^{-3}$, the search returns a gain of $-0.048$ at $\varepsilon = 0.10$ and $-0.085$ at $\varepsilon = 0.22$ as in Fig.~\ref{fig:gap}(b), in agreement to within $10^{-4}$ with $\tfrac{1}{2} + \tfrac{1}{2}\sum_b \sqrt{Q_b(1-Q_b)}$ evaluated at the edge of the tolerance. Absolute values in the untrusted rows are not comparable with those in the trusted rows, since the untrusted attacker is credited with the channel's noise as her own, and only the difference within a pair of rows is compared.

\begin{table}[t]
\caption{Basis asymmetry evaluated across two feasibility constraints and two noise conventions. Under the averaged constraint, the attacker can redistribute the disturbance between the bases, provided that $\FAB \geqslant 0.87$. Under the trusted per-basis constraint, each observed error rate must independently remain at or below $0.13$. These specific rows (marked $^\ast$) are evaluated using Eq.~\eqref{eq:composed} -- a result that the search independently reproduces in the averaged rows. Notably, at $\varepsilon = 0.22$, the one-basis channel already fails this per-basis test. Under the untrusted per-basis constraint, the attacker supplies the entire disturbance and must reproduce each observed rate to within a tolerance of $2\times10^{-3}$. Finally, absolute values should only be compared between rows that share the same noise convention.}
\label{tab:perbasis}
\begin{ruledtabular}
\setlength{\tabcolsep}{0pt}
\begin{tabular}{cccccc}
$\varepsilon$ & constraint & Noise & Symmetric $\FAE$ & One-basis $\FAE$ & Gain \\
\colrule
0.10 & averaged  & trusted   & 0.7846 & 0.7858 & $+0.0012$ \\
0.10 & per-basis & trusted$^\ast$ & 0.7846 & 0.7631 & $-0.0214$ \\
0.10 & per-basis & untrusted & 0.7220 & 0.6737 & $-0.0483$ \\
0.22 & averaged  & trusted   & 0.6581 & 0.6640 & $+0.0059$ \\
0.22 & per-basis & trusted$^\ast$ & 0.6581 & abort & --- \\
0.22 & per-basis & untrusted & 0.8153 & 0.7301 & $-0.0853$ \\
\end{tabular}
\end{ruledtabular}
\end{table}

The mechanism is the concavity of Eq.~\eqref{eq:composed}, whose two terms reward spreading a given amount of disturbance over both bases. An averaged constraint lets the attacker place disturbance wherever it is cheapest, so an asymmetric channel incurs no additional cost for her, and the masking effect gives her a small gain. A per-basis constraint ties each disturbance to the room left in its own basis, and a one-basis channel leaves that room unevenly: under the trusted convention the damaged basis admits only $D_Z \leqslant 0.0375$ at $\varepsilon = 0.10$, and under the untrusted convention the clean basis pins $Q_X \approx 0.002$, which caps her key-basis Holevo information at $h(Q_X) = 0.0204$ however much error the key basis carries (Sec.~\ref{sm:wlc}). The asymmetry advantage is therefore a property of the monitoring convention. Per-basis estimation is already the convention of security proofs, and the results quantify what an averaged threshold would concede on an asymmetric channel. At $\varepsilon = 0.10$ on the one-basis channel, replacing the averaged constraint by per-basis constraints at the same $13\%$ lowers the attainable fidelity from $0.7858$ to $0.7631$, by $0.023$, which requires no additional measurement, since the two rates are already collected.

\subsection{Non-Pauli noise and ancilla size}\label{sec:nonpauli}

The search is not restricted to Pauli noise. Amplitude damping, the process by which an excited state relaxes towards the ground state, is not a Pauli channel and lies outside the Pauli-channel cloning solutions of Refs.~\cite{cerf2000cloning, kerstan2026pauli}, and the same search treats it without modification. It returns circuits of four to five gates whose attainable leak declines from $0.820$ to $0.744$ as the damping parameter $\gamma$ runs from $0.05$ to $0.20$ (Table~\ref{tab:ad}). Enlarging the attacker's quantum memory at the one-basis operating point $\varepsilon = 0.10$ changes the attainable value by less than $10^{-3}$, from $0.7850$ with one ancilla qubit to $0.7858$ with two, and a third qubit produces no change at four-decimal resolution. The two-qubit ancilla used throughout is therefore adequate at that operating point, and since the ancilla size enters the search as an initialisation parameter, larger memories require no change to the solver.

\subsection{Bilateral noise in the entanglement-based protocol}\label{sec:di_channel}

In E91, the constraint becomes a margin on the CHSH value and the leakage measure is $\Holevo(A{:}E)$. The channel acts on both subsystems as $\Lambda_p^A \otimes \Lambda_p^B$, with $\Lambda_p$ the single-qubit depolarising channel of strength $p$. Depolarising noise multiplies every non-identity Pauli component of a state by $(1-p)$, and each CHSH correlator is bilinear in traceless Pauli operators on $A$ and $B$, so each correlator collects two factors,
\begin{equation}
S_{\mathrm{obs}}(U_E, p) = (1-p)^2\, S_{\mathrm{pre}}(U_E),
\label{eq:chsh}
\end{equation}
where $S_{\mathrm{pre}}(U_E)$ is the CHSH value of the state prepared by the attacker's interaction, $U_E$, before the channel acts. Tracing out Bob leaves only Pauli coefficients that act as the identity on $B$. Because $\Lambda_B$ does not damp these identity components, the eavesdropper's observed state remains independent of $\Lambda_B$ (so that $\rho_E^{\mathrm{obs}} = \rho_E^{\mathrm{pre}}$). This independence is confirmed numerically to four decimal places across $p \in \{0, 0.05, 0.1, 0.2, 0.5\}$. Consequently, the two subsystems constrain the attacker through distinctly separate mechanisms: $\Lambda_A$ lowers $\Holevo$ via data processing on Alice's measured qubit, whereas $\Lambda_B$ leaves $\Holevo$ unchanged but tightens the CHSH margin through Eq.~\eqref{eq:chsh}. Modelling only a single subsystem incorrectly degrades this margin by $(1-p)$ instead of $(1-p)^2$. At small $p$, this error halves the noise-induced loss of the CHSH value and estimates the failure point of the honest link at roughly twice its true noise level. Therefore, a single-subsystem noise model miscalculates the valid operating range by a factor of about two.

A second feature of the entanglement-based case bounds the attacker from the honest side. The asymptotic key rate ~\cite{acin2007diqkd}
\begin{equation}
    r(S) = 1 - h(Q) - h\!\left(\tfrac{1}{2} + \tfrac{1}{2}\sqrt{(S/2)^2 - 1}\right),
\end{equation}
with $Q$ the error rate between the measurements from which the key is formed, vanishes at $S_{\mathrm{sec}} \approx 2.4240$ when evaluated on the honest depolarised pair, for which $Q = \tfrac{1}{2}(1 - S/2\sqrt{2})$, and a finite-statistics buffer places the operational threshold at $\tau_S \approx 2.444$. Since the bilateral channel lowers the honest CHSH value to $(1-p)^2\,2\sqrt{2}$, an undisturbed Bell pair reaches $S_{\mathrm{sec}}$ at
\begin{equation}
p_{\mathrm{cliff}} = 1 - \sqrt{S_{\mathrm{sec}}/2\sqrt{2}} \approx 0.0742,
\label{eq:cliff}
\end{equation}
above which even an undisturbed channel yields no key, and it falls below $\tau_S$ at the slightly smaller noise level $p \approx 0.070$. The attacker therefore operates within a band that closes at a noise level fixed by quantum theory and by the key-rate formula, whereas in the prepare-and-measure case the user chooses the threshold independently of the channel. The extent to which a drifting channel can be exploited is accordingly determined by this geometry together with the abort threshold (Sec.~\ref{sec:thresholds}).

\section{Attacks under device drift}\label{sec:drift}

A drifting noise lower bound is typically treated as a benign operational error that merely reduces the key rate. Because the closed-form solutions of Sec.~\ref{intro_prior_work} and the relaxations used to certify key rates evaluate only a single, frozen channel state, neither can determine whether an adversary gains an advantage by adapting to the channel's drift. 

To address this, this section frames adaptive eavesdropping as a sequential decision problem (Sec.~\ref{sec:drift_formulation}). It constructs the libraries of quantum circuits from which the attacker selects their actions (Sec.~\ref{sec:libraries}), and describes the adaptive agent alongside its baselines (Sec.~\ref{sec:agent}). Finally, it quantifies the value of this adaptation for both BB84 (Sec.~\ref{sec:drift_bb84}) and E91 (Sec.~\ref{sec:drift_e91}), exploring its dependence on the threshold geometry (Sec.~\ref{sec:thresholds}).

\subsection{Adaptive eavesdropping as a constrained decision process}\label{sec:drift_formulation}

The channel parameter $p_t$, which is the bit-flip rate for BB84 and the depolarising strength for E91, follows an Ornstein--Uhlenbeck process~\cite{uhlenbeck1930theory},
\begin{equation}
\mathrm{d}p_t = \kappa\,(\mu - p_t)\,\mathrm{d}t + \sigma\,\mathrm{d}W_t,
\label{eq:ou}
\end{equation}
which reverts at rate $\kappa$ towards the mean $\mu$ under fluctuations of volatility $\sigma$ driven by a Wiener process $W_t$. It is integrated exactly with one round as the time step and clipped to the ranges of Table~\ref{tab:mapping}. An episode comprises $T = 50$ rounds. In each round the attacker selects one circuit $a_t$ from a library $\mathcal{A}$ and receives the reward $r(a_t, p_t)$, equal to $\FAE$ for BB84 and to $\Holevo(A{:}E)$ for E91. The single-round value of the monitored statistic, $\FAB$ for BB84 and $S$ for E91, is fixed by $a_t$ and $p_t$, since the statistic is evaluated asymptotically, and the legitimate parties abort when its mean over the episode falls below the threshold. In other words, the constraint is a budget over the entire episode and not a single round.

The attacker observes the current noise level $p_t$ together with its three previous values, her previous action, the mean of the monitored statistic over the last ten rounds, her mean reward so far and the number of rounds remaining. Knowledge of $p_t$ idealises an attacker who infers the channel's own noise from the public error-rate record, which she can do because she knows the disturbance she contributes. The idealisation removes the lag of that inference and so favours the attacker, which is the conservative direction for a defender. A policy $\pi$ maps observations to library circuits, and the attacker maximises
\begin{equation}
J_\lambda(\pi) = \mathbb{E}\!\left[\,\sum_{t=1}^{T} r(a_t, p_t)\right] - \lambda\,\Pr\!\left[D_\pi\right],
\label{eq:objective}
\end{equation}
where $D_\pi$ represents the event of detection (when the episode mean falls below the threshold) and $\lambda \geqslant 0$ is the detection penalty, tuned for either strict or tolerated settings. The results below report the mean reward per round alongside the observed detection frequency. Crucially, the compactness of the generated circuits allows us to build small, discrete libraries ($24$ circuits for BB84 and $26$ for E91, plus the identity, as detailed in Sec.~\ref{sec:libraries}). This reduces an otherwise intractable continuum of eavesdropping strategies into a finite action set, casting Eq.~\eqref{eq:objective} as a constrained Markov decision process solvable by standard methods.

\begin{table*}[t]
\caption{Component-by-component correspondence between the two drift environments. The monitored statistic is averaged over each episode of $T = 50$ rounds, and the drift is the discretised Ornstein--Uhlenbeck process of Eq.~\eqref{eq:ou} with one round as the time step. The cliff is the noise level at which the undisturbed channel alone reaches the abort threshold for BB84 and the zero-key-rate value $S_{\mathrm{sec}}$ for E91 (Eq.~\eqref{eq:cliff}).}
\label{tab:mapping}
\begin{ruledtabular}
\begin{tabular}{p{0.2\textwidth}p{0.37\textwidth}p{0.37\textwidth}}
Component & BB84 & E91/DIQKD \\
\colrule
Protocol family   & Prepare-and-measure & Entanglement-based, CHSH game \\
Channel           & Bit-flip of rate $p_{\mathrm{bf}}$, before the attacker's interaction & Bilateral depolarising $\Lambda_p^A \otimes \Lambda_p^B$, after the preparation \\
Attacker's memory & One qubit & One qubit (three-qubit preparation on $A$, $B$, $E$) \\
Honest reference  & $\FAB = 1 - p_{\mathrm{bf}}/2$ & $S_{\mathrm{obs}} = (1-p)^2 \cdot 2\sqrt{2}$ \\
Abort threshold (episode mean) & $\FAB \geqslant 0.87$ & $S_{\mathrm{obs}} \geqslant \tau_S \approx 2.444$ \\
Attacker's information & $\FAE$, memory measured in the announced basis & $\Holevo(A{:}E)$ \\
Drift parameters  & $\mu = 0.10$, $\sigma = 0.03$, $\kappa = 0.1$, $p \in [0.01, 0.30]$ & $\mu = 0.03$, $\sigma = 0.015$, $\kappa = 0.1$, $p \in [0, 0.0642]$ \\
Stationary spread $\sigma/\sqrt{2\kappa}$ & $0.067$ & $0.034$ \\
Library           & 24 circuits (4 levels $\times$ 6 anchors) $+$ identity & 26 circuits (7 levels $\times$ 4 anchors, 2 infeasible) $+$ identity \\
Cliff             & $p_{\mathrm{bf}} = 0.26$ & $p_{\mathrm{cliff}} \approx 0.0742$ \\
\end{tabular}
\end{ruledtabular}
\end{table*}

The contrast with a semidefinite relaxation is one of formulation. A relaxation characterises the correlations compatible with one fixed observed behaviour and answers a static question, namely how much the eavesdropper can know given that behaviour. In the adaptive problem the channel parameter evolves stochastically, so the constraint changes from round to round, the abort condition is a budget shared across the rounds of an episode, and the objective is an expectation over channel trajectories. Deciding how much of the budget to spend at the current noise level requires anticipating the noise levels to come, which no evaluation at a frozen behaviour supplies, so a semidefinite treatment would have to be embedded in a sequential optimisation that it does not itself provide.

Three reference policies bound the adaptive attacker. A \emph{fixed circuit} plays the same library element in every round of an episode. The \emph{greedy oracle} myopically selects the highest-reward circuit that keeps the ten-round mean above the threshold; it avoids detection but ignores future rounds. Finally, \emph{dynamic programming} solves Eq.~\eqref{eq:objective} over the augmented state $s_t = (t, p_t, E_t)$, where $E_t$ tracks the accumulated threshold excess. By applying a penalty strict enough to guarantee zero detection, its value, $V_{\mathrm{DP}}$, establishes the absolute performance upper bound for the library.

\begin{proposition}\label{prop:ordering}
Let $V_{\mathrm{fix}}$ be the mean reward per round of the best fixed circuit that incurs no detection, $V_\pi$ that of any policy $\pi$ that acts on the observations above and incurs no detection, and $V_{\mathrm{DP}}$ the value of dynamic programming over $s_t$ under the constraint that no trajectory ends in detection. Then, $V_{\mathrm{fix}} \leqslant V_\pi \leqslant V_{\mathrm{DP}}$.
\end{proposition}

\begin{proof}
A fixed circuit is a policy that ignores its observations, which gives the first inequality. The drift is Markov and $p_t$ is observed, and detection depends on the trajectory only through $E_T$, so $s_t$ is a sufficient statistic for the constrained problem. Every observation available to $\pi$ is a function of the history, and a history-dependent policy is dominated in a finite-horizon Markov decision process by the optimal policy over a sufficient state~\cite{puterman1994mdp}, which gives the second inequality.
\end{proof}

As feasible policies, both the greedy oracle and a zero-detection trained agent fall below this optimal upper bound. The agent's shortfall from $V_{\mathrm{DP}}$ reveals its unexploited potential, while the oracle's gap to $V_{\mathrm{DP}}$ quantifies the value of anticipating the episode budget.

\subsection{Libraries of learnt circuits}\label{sec:libraries}

The BB84 drift library was built before the stationary sweeps of Sec.~\ref{sec:stationary}, in the setting of Ref.~\cite{decker2024qkd}. The bit-flip acts on the travelling qubit before the attacker's interaction, so the noise also degrades her copy; she holds one ancilla qubit; and $\FAE$ is the probability that her ancilla, measured in the announced basis, returns Alice's bit, which lies below the Helstrom value of Eq.~\eqref{eq:helstrom}. The library comprises four levels of aggressiveness, each optimised at six noise anchors $p \in \{0.02, 0.05, 0.10, 0.15, 0.20, 0.25\}$, giving $24$ circuits (Sec.~\ref{sm:compute}). The drift results for BB84 therefore concern this weaker, noise-degraded attacker, and their absolute values are not comparable with those of Sec.~\ref{sec:noisy_channels}. The E91 library constrains the CHSH value alone, which follows the convention of the envelope it is compared against~\cite{acin2007diqkd} and is the weaker of the two possible constraints (the stronger one is the full behaviour $p(ab|xy)$, the set of outcome probabilities given the measurement settings). It spans seven absolute CHSH targets, $S_{\mathrm{target}} \in \{2.55, 2.50, 2.43, 2.30, 2.20, 2.00, 1.80\}$, and four noise anchors, $p \in \{0, 0.0214, 0.0428, 0.0642\}$. By Eq.~\eqref{eq:chsh}, no preparation yields a CHSH value above $s_{\mathrm{honest}}(p) = (1-p)^2\,2\sqrt{2}$, the value of the undisturbed channel, so cells with $S_{\mathrm{target}} > s_{\mathrm{honest}}(p) - 0.01$ are infeasible and are skipped, which leaves $26$ cells of five to nine gates. Each is built by the procedure of Sec.~\ref{sm:framework} on a three-qubit hardware-efficient preparation with fitness $\Holevo - \alpha(S - S_{\mathrm{target}})^2$ and $\alpha = 100$, which drives the cell to its prescribed CHSH value. The drift range for E91 is narrow by necessity, since its upper edge $p = 0.0642$ lies $0.01$ below the value $p_{\mathrm{cliff}}$ of Eq.~\eqref{eq:cliff}.

The library cells at $p = 0$ lie on the noiseless Ac\'in envelope (Fig.~\ref{fig:pareto}), which at $S = \tau_S$ permits $\Holevo \leqslant 0.607$, so the search saturates the analytical frontier where one exists. Averaged over the operating noise distribution, the attainable value is $\Holevo \approx 0.35$: the zero-detection upper bound over the library is $V_{\mathrm{DP}} \approx 0.356$, and a warm-started sequential-quadratic Pareto search over the attacker's full three-qubit pre-channel state, $16$ real parameters, returns $0.36$, so the library captures nearly all that this dimension-restricted attacker achieves. 

The separation between $0.35$ and $0.607$ does not indicate looseness of the envelope for two reasons. The envelope is evaluated at $p = 0$ whereas the attainable value is averaged over $p > 0$, so a noiseless upper bound is placed beside a noisy lower bound, and the bilateral channel of Eq.~\eqref{eq:chsh} accounts for most of the difference. In addition, the attacks are constructed on a fixed three-qubit preparation, whereas a device-independent statement admits an eavesdropper of unbounded dimension, so the value is a constructive lower bound on what a dimension-restricted attacker achieves at operating noise. Section~\ref{sm:keyrates} places both numbers beside the upper bound returned by a semidefinite relaxation.

\begin{figure}[t]
\centering
\includegraphics[width=\linewidth]{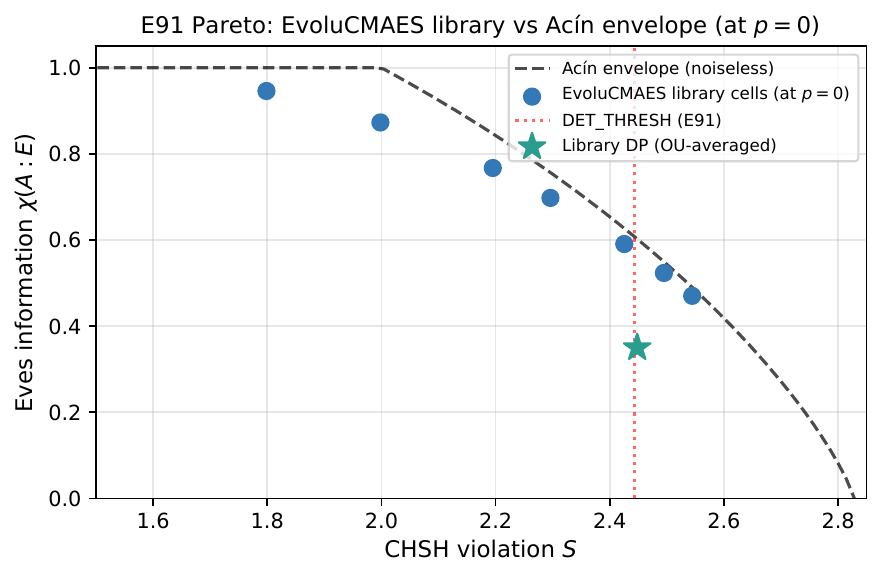}
\caption{Disturbance versus information in the device-independent setting. At $p = 0$, the discovered attacks trace the noiseless Ac\'in envelope~\cite{acin2007diqkd, pironio2009diqkd}, demonstrating that the search successfully saturates the analytical frontier where one exists. However, when averaged over the drifting noise, the attainable value (marked by the star) drops to $\Holevo \approx 0.35$, falling short of the $0.607$ permitted by the theoretical envelope at the threshold. This gap occurs for two reasons: the theoretical envelope assumes a strictly noiseless channel ($p = 0$) while the attainable value averages over active noise ($p > 0$), and the constructed attacks are restricted to a fixed three-qubit preparation.}
\label{fig:pareto}
\end{figure}

\subsection{Adaptive agent and baselines}\label{sec:agent}

The adaptive attacker is a value-based reinforcement-learning agent featuring double and duelling value estimation, $n$-step returns, and prioritised replay~\cite{hessel2018rainbow}. It employs invalid-action masking~\cite{huang2022maskableppo} to dynamically exclude circuits that would drop the ten-round mean below the threshold. Sharing this filter with the greedy oracle, the agent must optimally select among the remaining safe actions to maximise future rewards. 

Hyperparameters are identical across both protocols, which differ only by channel, threshold, and information measure. Training applies either a strict (heavy) or tolerated (moderate) detection penalty. As is standard, results report the best evaluation checkpoint, with variance detailed in Sec.~\ref{sm:drift}. Baselines include a passive attacker, the reference policies defined above, and for E91, a masked policy-gradient agent~\cite{schulman2017ppo, huang2022maskableppo} alongside behavioural cloning of the dynamic programming policy~\cite{ross2011dagger}. Unmasked policy-gradient agents are compared in Figs.~\ref{fig:bb84methods} and~\ref{fig:e91methods}.

\subsection{Value of adaptation on BB84}\label{sec:drift_bb84}

As detailed in Table~\ref{tab:results}, the strict adaptive agent on BB84 reaches $\FAE = 0.7434$ with zero detection. This matches the greedy oracle and achieves $98.8\%$ of the zero-detection upper bound, $V_{\mathrm{DP}} = 0.7527$ (Figs.~\ref{fig:bb84} and~\ref{fig:bb84methods}). The ``static'' baseline in the released environment actually adapts to drift by selecting the most conservative library level nearest the current noise anchor, reaching $\FAE = 0.7197$. The adaptive agent outperforms this conservative rule by $0.024$, capturing roughly three-quarters of the theoretical $0.033$ headroom available up to $V_{\mathrm{DP}}$. The remaining $0.009$ gap represents the value of anticipating the episode budget. The best genuinely fixed circuit scores below $0.7197$. Finally, unmasked policy-gradient agents plateau roughly $0.02$ below the oracle.

\begin{table}[t]
\caption{Value of adaptation. Mean information per round and detection rate, alongside percentages of the zero-detection dynamic-programming upper bound, $V_{\mathrm{DP}}$, evaluated over the same library. BB84 values are averaged over $50$ evaluation episodes. E91 agent and oracle values are averaged over $200$ episodes, while its fixed-circuit and dynamic-programming baselines are evaluated over $2000$ episodes. The E91 upper bound is resolved computationally to a grid precision of $0.001$ (Sec.~\ref{sm:drift}). The BB84 conservative lookup policy is not a fixed circuit; it dynamically selects the safest library level anchored nearest the current noise level.}
\label{tab:results}
\begin{ruledtabular}
\setlength{\tabcolsep}{3pt}
\begin{tabular}{llccc}
 & Policy & $\FAE/\Holevo$ & Det. & \% of $V_{\mathrm{DP}}$ \\
\colrule
\multirow{5}{*}{BB84}
 & Conservative lookup & 0.7197 & 0\% & 96 \\
 & Greedy oracle             & 0.7434 & 0\% & 99 \\
 & Adaptive, strict          & 0.7434 & 0\% & 99 \\
 & Adaptive, tolerated       & 0.7477 & 4\% & 99 \\
 & Dyn.\ prog., $V_{\mathrm{DP}}$ & 0.7527 & 0\% & 100 \\
\colrule
\multirow{5}{*}{E91}
 & Best fixed circuit & 0.135 & 0\% & 38 \\
 & Greedy oracle             & 0.3507 & 0\% & 99 \\
 & Adaptive                  & 0.3484 & 0\% & 98 \\
 & Dyn.\ prog., coarse & 0.3540 & 7\% & 99 \\
 & Dyn.\ prog., $V_{\mathrm{DP}}$ & 0.356 & 0\% & 100 \\
\end{tabular}
\end{ruledtabular}
\end{table}

This $0.024$ separation between the conservative rule and the adaptive agent quantifies the impact of drift on BB84. A threat model that confines the attacker to the safest circuit at each noise level undercounts her information by $0.024$ in $\FAE$, and assuming a single fixed circuit undercounts it even further. Furthermore, these drift experiments enforce the basis-averaged constraint $\FAB \geqslant 0.87$. Because per-basis monitoring substantially alters the attacker's attainable information (Sec.~\ref{sec:asymmetry}), evaluating this separation under per-basis constraints is left for future measurements.

\begin{figure}[t]
\centering
\includegraphics[width=\linewidth]{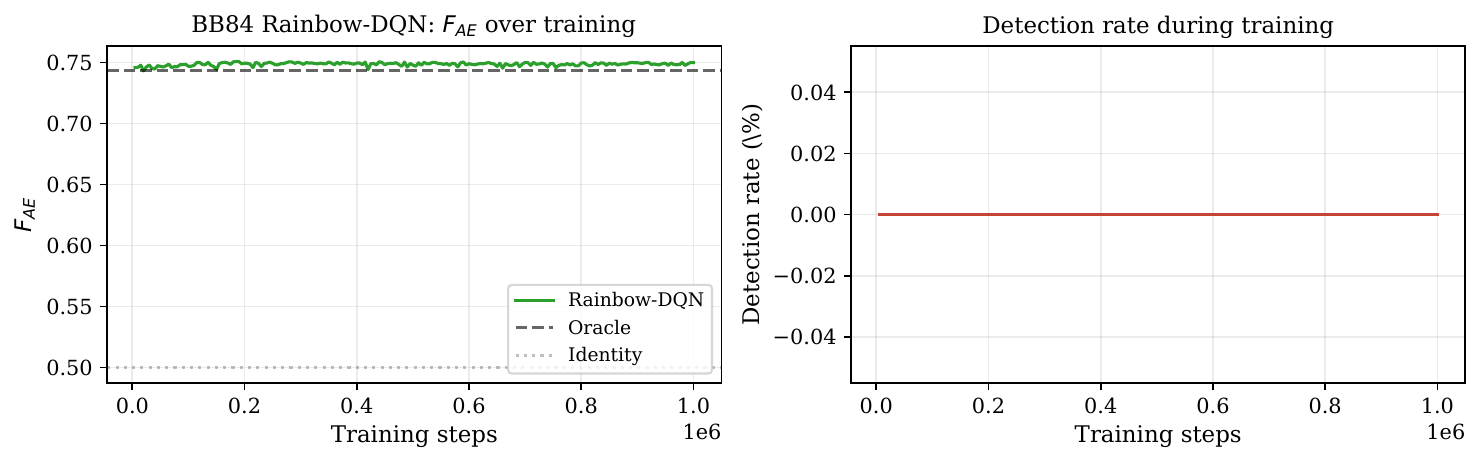}
\caption{Adaptive attack on BB84 under a drifting device. Upper panel: attacker information $\FAE$ during evaluation, reaching the greedy oracle (dashed) at its best checkpoints and well above the passive baseline (dotted). Lower panel: detection rate, which stays at zero at those checkpoints.}
\label{fig:bb84}
\end{figure}

\begin{figure}[t]
\centering
\includegraphics[width=\linewidth]{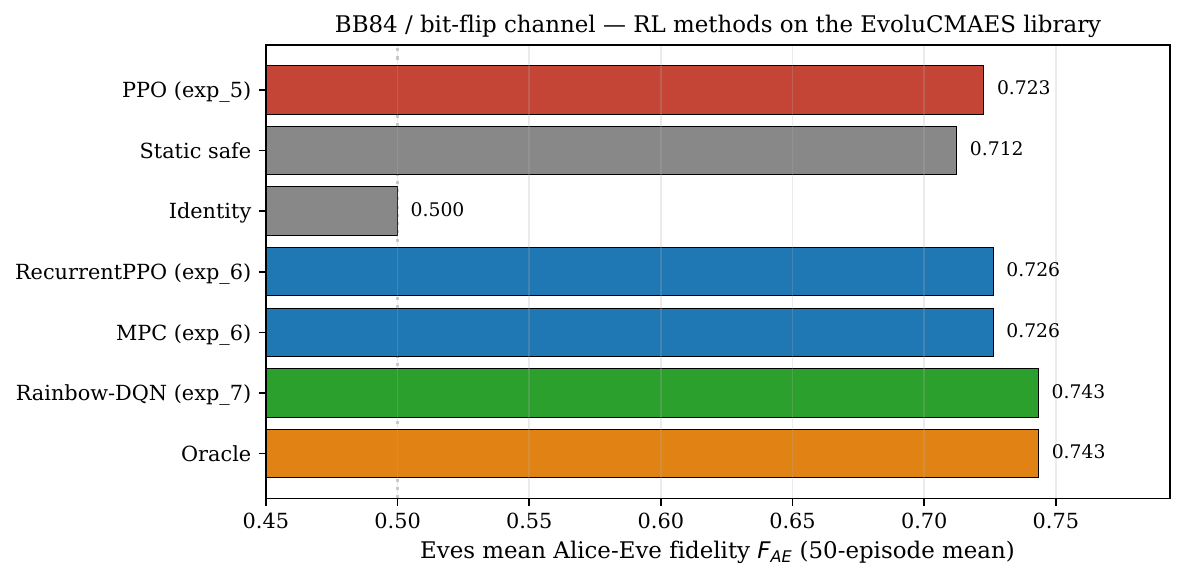}
\caption{BB84 under a drifting bit-flip channel: final evaluation $\FAE$ across policies, strict-feasibility variants only. The value-based agent matches the greedy oracle at $0.7434$, unmasked policy-gradient methods plateau around $0.02$ below it, and the noise-indexed conservative lookup (``Static safe'') reaches $0.7197$.}
\label{fig:bb84methods}
\end{figure}

\subsection{Value of adaptation on E91}\label{sec:drift_e91}

On E91, adaptation provides a massive advantage. To survive all $2000$ evaluation episodes without detection, the best fixed circuit must conservatively keep the episode mean of $S$ above $\tau_S$ even on the noisiest trajectories, yielding only $\Holevo = 0.135$. In contrast, the adaptive agent reaches $0.3484$ at zero detection -- within $2\%$ of the $V_{\mathrm{DP}} \approx 0.356$ upper bound (Figs.~\ref{fig:e91training} and~\ref{fig:e91methods}). Adaptation thus multiplies the attacker's extracted information by approximately $2.6$ without increasing detection risk.

The agent's slight $0.0023$ shortfall compared to the greedy oracle requires paired statistical analysis. Because all E91 policies evaluate the same per-episode channel realisations, paired observations eliminate the channel-to-channel variance that would otherwise obscure such small effect sizes. Across $n = 200$ paired episodes, the agent trails the greedy oracle by $0.0023$ ($95\%$ confidence interval $[0.0018, 0.0028]$) and outperforms the masked policy-gradient agent by $0.0017$ ($[0.0012, 0.0022]$). Both differences are highly significant under a two-sided Wilcoxon signed-rank test~\cite{wilcoxon1945individual} (Table~\ref{tab:stats}). Meanwhile, unmasked policy-gradient agents collapse to $\Holevo \approx 0.26$--$0.28$, failing to explore effectively and confining themselves to just three or four distinct actions.

Finally, the released dynamic-programming experiments reach $0.3540$ but incur a $7\%$ detection rate. This detection is strictly an artefact of discretising the accumulated statistic onto a coarse grid. Recalculating on a fine grid yields a zero-detection optimum of $0.356 \pm 0.001$ (Sec.~\ref{sm:drift}). Because this fine-grid upper bound sits above the coarse run, it confirms that dynamic programming outperforms the agent through optimal decision-making, not by trivially tolerating detection.

\begin{table*}[t]
\caption{Paired per-episode comparisons on E91, $n = 200$ episodes with shared channel seeds. Two-sided Wilcoxon signed-rank $p$-values~\cite{wilcoxon1945individual}.}
\label{tab:stats}
\begin{ruledtabular}
\begin{tabular}{lccc}
Comparison & Mean difference & $95\%$ CI & Wilcoxon $p$ \\
\colrule
Adaptive $-$ greedy oracle                  & $-0.0023$ & $[-0.0028, -0.0018]$ & $2.5 \times 10^{-15}$ \\
Adaptive $-$ masked policy gradient         & $+0.0017$ & $[+0.0012, +0.0022]$ & $1.0 \times 10^{-8}$ \\
Adaptive $-$ dynamic programming (7\% det.) & $-0.0056$ & $[-0.0065, -0.0047]$ & $1.3 \times 10^{-21}$ \\
\end{tabular}
\end{ruledtabular}
\end{table*}

\begin{figure}[t]
\centering
\includegraphics[width=\linewidth]{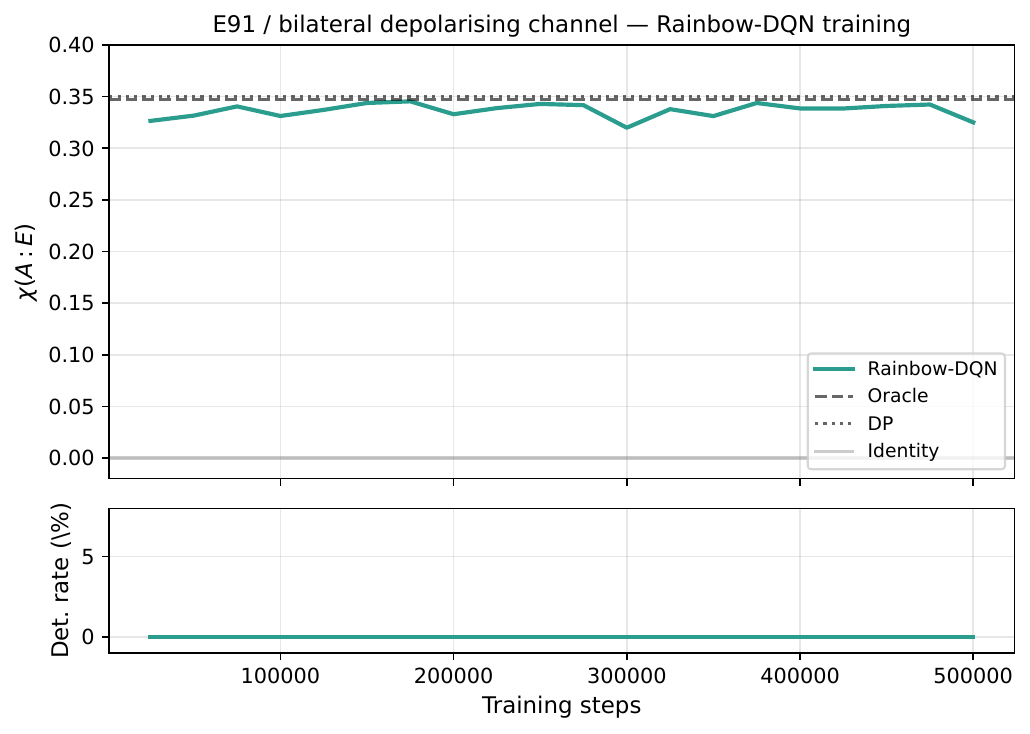}
\caption{E91 under the bilateral depolarising channel: training curve for the adaptive attack. Upper panel, Holevo information during evaluation against the greedy oracle (dashed), the coarse-grid dynamic-programming run (dotted) and the passive baseline (grey); lower panel, detection rate.}
\label{fig:e91training}
\end{figure}

\begin{figure}[t]
\centering
\includegraphics[width=\linewidth]{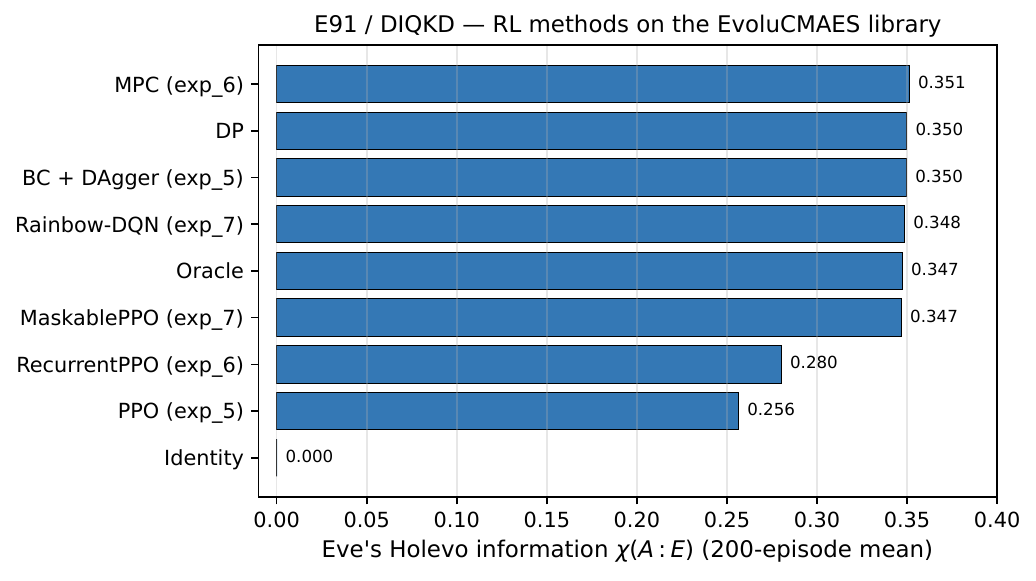}
\caption{E91/DIQKD: final evaluation $\Holevo$ across policies. The masked value-based agent reaches $0.3484$ against a greedy oracle at $0.3507$ and a coarse-grid dynamic-programming run at $0.3540$ with $7\%$ detection, whereas unmasked policy-gradient methods collapse to $\Holevo \approx 0.26$--$0.28$ by confining themselves to three or four distinct actions.}
\label{fig:e91methods}
\end{figure}

\subsection{Tolerance for detection}\label{sec:thresholds}

On BB84 the agent trained with the tolerated penalty reaches $\FAE = 0.7477$ at a $4\%$ detection rate, $0.0043$ above the strict agent. Since $V_{\mathrm{DP}} = 0.7527$ is reached without detection, accepting detection is not required to obtain that information; the gain reflects the easier learning problem posed by a softer penalty and not a trade-off inherent to the task. On E91 the tolerated agent coincides with the strict one. The drift range there stops $0.006$ short of $p \approx 0.070$, where the honest CHSH value falls to $\tau_S$, just below the noise level $p_{\mathrm{cliff}}$ of Eq.~\eqref{eq:cliff} that quantum theory and the key-rate formula fix, whereas on BB84 the user chooses the threshold independently of the channel. It remains untested whether this difference in threshold geometry dictates how an attacker trades detection risk for information gain. Sweeping the detection penalty within a fine-grid dynamic programme would resolve this question.

\section{Discussion}\label{sec:discussion}

Explicit attacks serve the defender by providing empirically testable provisioning margins, whereas margins calibrated against conjectured strategies cannot be directly verified. However, this framework is designed to complement, rather than replace, formal security proofs. Because the search identifies feasible points without providing formal certificates, it yields constructive lower bounds on the eavesdropper's information. Semidefinite relaxations and entropic bounds remain the appropriate instruments for proving rigorous upper bounds on the key rate. 

The adaptive threat models evaluated here are deliberately scoped. Our values assume asymptotic detection statistics and an attacker who perfectly tracks the current noise level. The impact of delayed noise estimation or finite-sample fluctuations on the adaptive advantage remains an open question for future work. Furthermore, the specific BB84 drift library evaluated here restricts the attacker to a one-qubit memory and imposes a basis-averaged constraint. While Proposition~\ref{prop:ordering} structurally accommodates per-basis constraints, Sec.~\ref{sec:asymmetry} demonstrates that per-basis monitoring alters the attacker's attainable information, making the quantification of this separation an important next step.

Finally, establishing constructive bounds is established practice, validating theoretical limits such as the device-independent collective-attack rate~\cite{acin2007diqkd, pironio2009diqkd} and device-dependent rates for imperfect sources~\cite{woodhead2014tight}. Our analysis aligns with standard theoretical conventions: it assumes flawless physical implementation (leaving existing side-channel security proofs intact), operates strictly at the qubit level, and weights the four sifted states equally, corresponding to an unbiased basis choice.

\section{Practical consequences}\label{sec:practical}

Historically, device drift has entered the provisioning of a QKD link merely as a source of error that reduces the secure key rate. However, under the trusted-noise threat model, our measurements demonstrate that drift also dynamically dictates the amount of information an attacker can extract. For device-independent E91, an attacker adapting across $26$ compact circuits holds $2.6$ times the Holevo information of the best single circuit at zero detection. On BB84, she outperforms a conservative noise-indexed rule by $0.024$ in fidelity.

A committed attacker is constrained by the noisiest stretch of the channel she must traverse undetected. When this operational requirement spans a full recalibration interval, the relevant excursion is the largest one within that window -- a value that scales with the interval's length for an Ornstein--Uhlenbeck process. Consequently, the recalibration interval heavily impacts the attacker's attainable information, a vulnerability that static analyses overlook.

The chosen monitoring convention exerts a comparable impact. On a stationary channel concentrating its error in one basis at $\varepsilon = 0.10$, enforcing individual per-basis thresholds rather than an averaged threshold lowers the attacker's attainable fidelity by $0.023$ and reverses her gain from basis asymmetry -- all without requiring additional measurements by the legitimate parties. 

Ultimately, this computational sequential-decision apparatus applies to operating points that lack closed-form analyses, including alternative drift processes and finite-key regimes constrained by confidence intervals. Wherever a formal security proof establishes an upper bound on an eavesdropper's information, explicit adaptive constructions of this kind provide the rigorous lower bound beneath it.

\begin{acknowledgments}
The authors thank N. L\"utkenhaus for advice on security models and P. Brown for connecting the work to upper bounds on key rates. Code, analysis notebooks and the channel sweeps reported here are available upon request and will be released under the MIT licence at the end of peer review. Generative AI tools assisted with implementing experiments and editing the text.
\end{acknowledgments}

\bibliographystyle{apsrev4-2}
\bibliography{references}

\clearpage
\onecolumngrid

\begin{center}
{\large \textbf{Supplemental Material}}\\[2pt]
{\textit{Learnt Attacks on Quantum Key Distribution under Channel Noise and Device Drift}}\\[4pt]
Marcel Mordarski, Daniel Budina, Benjamin Gras, Abdelrahman Shehata, and Roberto Bondesan
\end{center}

\vspace{1em}

\setcounter{section}{0}
\setcounter{equation}{0}
\setcounter{figure}{0}
\setcounter{table}{0}
\setcounter{secnumdepth}{2}

\renewcommand{\theequation}{S\arabic{equation}}
\renewcommand{\thefigure}{S\arabic{figure}}
\renewcommand{\thetable}{S\arabic{table}}
\renewcommand{\thesection}{S\arabic{section}}
\renewcommand{\thesubsection}{\thesection.\arabic{subsection}}

\makeatletter
\renewcommand{\p@subsection}{}
\makeatother

\section{Learnt attacks on a fixed template: a review of Ref.~\cite{decker2024qkd}}
\label{sm:decker}

As the present work extends Ref.~\cite{decker2024qkd} directly, that work is summarised here in the notation of the main text.

That work recasts eavesdropping on BB84 as a variational learning problem. Alice sends one of the four BB84 states, and the attacker applies a unitary $U(\Theta)$ jointly to the travelling qubit and an ancilla she keeps, forwards the travelling qubit to Bob, and after the public basis announcement applies a second unitary $V(\Lambda)$ to her memory before measuring it. Both unitaries come from a hardware-efficient template, a fixed alternation of single-qubit rotations and \textsc{cnot}s whose depth and connectivity are chosen in advance. Only the rotation angles $(\Theta, \Lambda)$ vary during training, which proceeds by Adam~\cite{kingma2014adam} on the loss of Eq.~\eqref{eq:decker} of the main text, the quadratic term pinning $\FAB$ to a chosen target $f$ and the linear term maximising $\FAE$. Fidelities follow $F(\rho_A, \rho_B) = \left(\Tr\sqrt{\sqrt{\rho_A}\,\rho_B\sqrt{\rho_A}}\right)^2$, with $\FAB$ averaged over the four sifted states.

Three findings are relevant here. Firstly, on an ideal channel the learner converges to the phase-covariant cloning machine~\cite{bruss2000phasecov}, which validates the variational apparatus against a known answer. Secondly, adding a bit-flip to the channel, applied before the attacker's interaction in that work, breaks that cloner's optimality, since an imbalanced construction with cloning coefficients $C_{AB,Z} = \sin\psi$ and $C_{AB,X} = \cos\phi$, carrying two free parameters in place of one, overtakes it whenever the channel's Kraus weights satisfy $\alpha\gamma \neq \beta\delta$ in the notation of Ref.~\cite{decker2024qkd}, giving up fidelity in the damaged basis for fidelity in the untouched one. Thirdly, treating the parity information leaked during Cascade reconciliation~\cite{brassard1993secret} as available to the attacker permits a joint measurement over stored ancilla pairs, which raises her recovery probability from $0.810$ to the Helstrom optimum $0.894$ at $\FAB = 0.892$, the $11\%$ error bound~\cite{helstrom1969quantum}.

The present work retains the threat model and the figures of merit, with the fidelities evaluated as the recovery probabilities defined in Sec.~\ref{sec:setting}, and modifies two elements. The circuit template is removed, so that gate structure is searched (Sec.~\ref{sm:framework}), and the single stationary noise level is replaced by a sweep over channel families (Secs.~\ref{sec:stationary} and~\ref{sm:sweeps}) and by an explicitly time-varying channel (Sec.~\ref{sec:drift}). Reference~\cite{decker2024qkd} demonstrated particular departures from the textbook attack at one noise level. The objective here is different, namely to measure what an attacker gains by adapting to a channel that moves, and the learnt attacks enter as the actions of that sequential problem. That use requires an optimiser that does not presuppose the structure of the solution and a library of solutions small enough to select from round by round.

\section{The two-loop search procedure}
\label{sm:framework}

The outer loop maintains a population of $\{T_i, \theta_i\}$ pairs. Individuals are ranked by fitness each generation, the top-ranked are retained under elitism, and the individuals retained by the selection step are perturbed by a topology-mutation operator that adds, removes or replaces one gate, with addition and removal matched so that the operator carries no net growth bias. The inner loop optimises the angles of each new or modified topology by a gradient-free routine at a fixed budget of $B$ fitness evaluations, covariance-matrix adaptation~\cite{Hansen2001, hansen2016cmaes, hansen2003cmaesreduced} for the BB84 sweeps and simultaneous-perturbation stochastic approximation~\cite{spall1992multivariate} for the E91 library. Each query is a density-matrix simulation of a non-convex objective without available gradients, which favours this class of inner optimiser~\cite{hansen2003cmaesreduced}; Table~\ref{tab:inner} confirms that the two routines, and uniform random sampling, return the same fidelity and gate count on a shared task.

The gate alphabet, the fitness and the inner optimiser are the only elements that change between the protocols. For BB84 the fitness is $\FAE$ subject to $\FAB \geqslant 0.87$ enforced as a feasibility mask, with covariance-matrix adaptation on the angles; for E91 it is $\Holevo - \alpha(S - S_{\mathrm{target}})^2$ with $\alpha = 100$ on a three-qubit hardware-efficient preparation, with simultaneous-perturbation stochastic approximation on the angles, which drives each library cell to a prescribed CHSH value. The population dynamics, selection and elitism of the outer loop are shared, and Table~\ref{tab:inner} shows that the inner optimiser is immaterial at fixed budget, the two reaching the same $\FAE$ and the same equilibrium gate count on the BB84 cloning task.

The evaluation budget also sets how large the returned circuits are, which keeps the library small enough to serve as a discrete action set in Sec.~\ref{sec:drift}. Enlarging $T$ at fixed $B$ enlarges $\dim(\theta)$, so the inner loop covers more dimensions with the same number of queries, and the fitness estimate returned to the outer loop is noisier for larger topologies. Rank-based selection comparing noisy estimates then favours the smaller of two topologies sharing a noiseless optimum, and the population settles below the dense baseline without an explicit size penalty. The size at which it settles is measured (Table~\ref{tab:inner}) and not derived; the mechanism suggests that it is set by the size at which the variance of the fitness estimate at budget $B$ becomes comparable to the fitness separation that selection must resolve. Circuits from the exact-evaluation sweeps of Sec.~\ref{sm:sweeps} run slightly larger, four to eight gates against four to six, since exact evaluation removes the selection pressure towards small circuits.

Table~\ref{tab:inner} reports the BB84 cloning task under three inner-loop optimisers with the outer loop, mutation operator and selection rule held fixed, with one fitness query counting as one unit for every method. Random sampling receives twice the budget to offset its lack of adaptation.

\begin{table}[htbp]
\caption{Inner-loop optimiser comparison on the BB84 cloning task; mean $\pm$ standard deviation over five seeds. All three recover $\FAE$ to within $0.01$ of one another and reach the same equilibrium of four to six gates, so the discovered circuit size is set by the budget and is insensitive to the choice of algorithm.}
\label{tab:inner}
\begin{ruledtabular}
\begin{tabular}{lccccc}
Inner optimiser & $\FAE$ & $\FAB$ & Gates & Evaluations & Best gates \\
\colrule
CMA-ES                          & $0.764 \pm 0.001$ & $0.794 \pm 0.000$ & $5.2 \pm 0.7$ & 2000 & 4 \\
Stochastic perturbation (SPSA) \cite{spall1992multivariate}  & $0.760 \pm 0.003$ & $0.794 \pm 0.002$ & $6.4 \pm 0.8$ & 2000 & 6 \\
Uniform random                  & $0.756 \pm 0.006$ & $0.796 \pm 0.002$ & $5.6 \pm 1.0$ & 4000 & 6 \\
\end{tabular}
\end{ruledtabular}
\end{table}

\section{Validation against known closed forms}
\label{sm:validation}

This section collects the controls that fix the accuracy of the search against analytically known answers. None of the closed forms below is supplied to the search, which starts from random gate sequences in every run, and the results validate the apparatus used in the main text.

\subsection{Reproduction of analytical cloners}\label{sec:analytical_cloners}

On a noiseless channel the optimal individual attack satisfies $\FAE = \tfrac{1}{2} + \sqrt{\FAB(1-\FAB)}$, the one-angle family of Refs.~\cite{fuchs1997optimal, cerf2000cloning}, which balances at $\FAB = \FAE \approx 0.853$. The expression holds along the whole family, so the thresholds tested lie on one analytically known curve and the exercise constitutes one control repeated at four abscissae. Starting from random gate sequences, the search saturates the envelope to four decimal places at every threshold tested, attaining $100.00\%$ of the analytical value at $\FAB \in \{0.95, 0.90, 0.87, 0.85\}$ (Table~\ref{tab:validation}). The objective contains no gate-count penalty, and the solutions nonetheless occupy four to eight gates drawn from a space admitting arbitrary depth, for the reason given in Sec.~\ref{sm:framework}.

\begin{table}[htbp]
\caption{Validation on the noiseless symmetric channel. The search receives no analytical input and recovers the closed form $\FAE = \tfrac{1}{2} + \sqrt{\FAB(1-\FAB)}$ at every threshold tested.}
\label{tab:validation}
\begin{ruledtabular}
\begin{tabular}{ccccc}
Threshold $\tau$ & $\FAB$ attained & $\FAE$ attained & Closed form & Ratio \\
\colrule
0.95 & 0.9500 & 0.7179 & 0.7179 & $1.0000$ \\
0.90 & 0.9000 & 0.8000 & 0.8000 & $1.0000$ \\
0.87 & 0.8700 & 0.8363 & 0.8363 & $1.0000$ \\
0.85 & 0.8500 & 0.8571 & 0.8571 & $1.0000$ \\
\end{tabular}
\end{ruledtabular}
\end{table}

\subsection{Collective attacks and the asymptotic key rate}\label{sec:key_rate_bounds}

Any attack consistent with the observed statistics is a valid eavesdropping strategy, so the key rate that remains under a specific attack bounds the universally secure rate from above, $r \leqslant 1 - h(Q) - \Holevo(A{:}E)$ by the Devetak--Winter relation~\cite{devetak2005distillation} (Sec.~\ref{sm:keyrates}). Running the search to maximise $\Holevo$ in place of $\FAE$ on the symmetric channel, whose quantum bit error rate is $Q$ in each basis, returns $\Holevo = h(Q)$ to four decimal places at $Q \in \{0.02, 0.04, 0.06, 0.11\}$, hence $r \leqslant 1 - 2h(Q)$ with a residual below $10^{-4}$ and the rate vanishing at $Q = 0.11$ (Table~\ref{tab:keyrate}). The rate $1 - 2h(Q)$ is the symmetric case of the Shor--Preskill rate~\cite{shor2000simple}, which the construction reproduces from above without having been supplied with it.

\begin{table}[htbp]
\caption{Key-rate upper bounds from constructed attacks on the symmetric channel. The search maximises $\Holevo$ behind the same feasibility constraint and is given no analytical input; it returns $\Holevo = h(Q)$ at every $Q$ tested, so the upper bound it produces coincides with the symmetric Shor--Preskill rate $1 - 2h(Q)$~\cite{shor2000simple} to within $10^{-4}$.}
\label{tab:keyrate}
\begin{ruledtabular}
\setlength{\tabcolsep}{4pt}
\begin{tabular}{ccccc}
$Q$ & $\Holevo$ attained & $h(Q)$ & $r \leqslant 1-h(Q)-\Holevo$ & $1-2h(Q)$ \\
\colrule
0.02 & 0.1414 & 0.1414 & 0.7171 & 0.7171 \\
0.04 & 0.2423 & 0.2423 & 0.5155 & 0.5154 \\
0.06 & 0.3274 & 0.3274 & 0.3451 & 0.3451 \\
0.11 & 0.4999 & 0.4999 & 0.0002 & 0.0002 \\
\end{tabular}
\end{ruledtabular}
\end{table}

\subsection{Comparison with a reliable numerical lower bound}\label{sec:benchmarking}

The natural test of a key-rate upper bound is whether it meets the lower bound that reliable numerical methods supply. Winick, L\"utkenhaus and Coles~\cite{winick2018reliable} bound the device-dependent rate from below by minimising the relative entropy between the post-interaction state and its pinching in the key basis, over every state compatible with the observed statistics. Given the same input, namely the per-basis error rates $(Q_Z, Q_X)$ of a BB84 round with the eavesdropper held responsible for all of the disturbance, the two methods bound the same quantity from opposite sides. Across four operating points running from the symmetric channel to one carrying almost all of its error in a single basis, the constructed attack extracts $\Holevo(A_Z{:}E) = h(Q_X)$, and the upper bound it certifies agrees with the numerical lower bound to within $3\times10^{-4}$ in rate (Table~\ref{tab:wlc-full}). Both bounds reproduce the Shor--Preskill rate $1 - h(Q_Z) - h(Q_X)$~\cite{shor2000simple}, which the Bell-diagonal reduction of Sec.~\ref{sm:wlc} returns analytically and which an eavesdropper holding a purification of the minimising state attains, so the agreement validates both methods against a known answer along the whole asymmetry axis.

In a first pass, two of the four operating points yielded upper bounds above the numerical lower bound by $0.07$ and $0.04$ in rate. Because the lower bound is a rigorous bound, the discrepancy could lie only in the search, and raising the evaluation budget from $400$ to $900$ evaluations per candidate and the outer loop from $18$ to $24$ generations closed both gaps to below $10^{-4}$ and returned smaller circuits. Computing an upper bound and a lower bound through independent methods therefore yields a convergence diagnostic that neither approach offers alone, and Sec.~\ref{sm:wlc} details the comparison.

\section{Channel sweeps: specification and results}
\label{sm:sweeps}

These sweeps use exact density-matrix evaluation, so the frontier they trace carries no estimator noise. Alice's four sifted states are prepared in turn, the attacker's unitary acts on the travelling qubit together with a two-qubit ancilla unless stated otherwise, and the channel then acts on the travelling qubit alone, leaving the attacker's reduced state untouched, so that the two figures of merit decouple and follow from the two marginals separately. Bob's fidelity is the mean over the four states of his recovery probability in the matched basis, and the attacker's fidelity is the mean over the two bases of the Helstrom success probability of Eq.~\eqref{eq:helstrom}. Each configuration runs from three seeds at population $14$, $18$ generations, elitism $4$ and inner budget $B = 350$, and the best feasible run is reported. The attainable values agree with the closed form of Eq.~\eqref{eq:composed} to within $2\times10^{-4}$ at every entry, the largest deviation being a shortfall of $2\times10^{-4}$ on the symmetric channel at $\varepsilon = 0.06$, where the closed form gives $0.8083$.

\begin{table}[htbp]
\caption{Attainable $\FAE$ against total honest channel error $\varepsilon$ at threshold $\FAB \geqslant 0.87$, for a channel placing all error in one basis ($p_x = \varepsilon$) and one splitting it evenly ($p_x = p_z = \varepsilon/2$). Both present honest fidelity $1 - \varepsilon/2$, so the noiseless closed form evaluated at the observed $\FAB$ assigns both the same value, $0.8363$. The final column is the excess of that closed form over the attainable value under the trusted-noise convention, and measures the conservatism of the closed form for a physically restricted attacker.}
\label{tab:sweep}
\begin{ruledtabular}
\begin{tabular}{cccccc}
$\varepsilon$ & Honest $\FAB$ & Symmetric & One-basis & Asymmetry gain & Closed-form excess \\
\colrule
0.02 & 0.990 & 0.8278 & 0.8278 & $<0.0001$ & $+0.0085$ \\
0.06 & 0.970 & 0.8081 & 0.8087 & $+0.0006$ & $+0.0276$ \\
0.10 & 0.950 & 0.7846 & 0.7858 & $+0.0012$ & $+0.0505$ \\
0.16 & 0.920 & 0.7366 & 0.7400 & $+0.0034$ & $+0.0963$ \\
0.22 & 0.890 & 0.6581 & 0.6640 & $+0.0059$ & $+0.1723$ \\
\end{tabular}
\end{ruledtabular}
\end{table}

Sweeping the split of a fixed $\varepsilon = 0.10$ between the bases in six steps moves the attainable value monotonically from $0.7846$ to $0.7858$, with a cross-seed spread near $0.001$ on the best-of-three statistic. Transplanting the circuit optimised on the symmetric channel onto the one-basis channel leaves it feasible at $\FAB = 0.8700$ and yields $\FAE = 0.7846$, short of the re-optimised $0.7858$ by $0.0012$.

\subsection{Per-basis repetition}
\label{sm:perbasis}

Every value in Table~\ref{tab:sweep} is obtained under a constraint on the basis-averaged fidelity, $\FAB \geqslant 0.87$, with the channel noise trusted. The untrusted per-basis rows of Table~\ref{tab:perbasis} in the main text use the same search configuration with the constraint applied to both rates independently and the attacker held responsible for the whole disturbance, so that her circuit must reproduce the observed per-basis rates $Q_Z = p_x$ and $Q_X = p_z$ to within $2\times10^{-3}$. The values returned coincide to within $10^{-4}$ with $\tfrac{1}{2} + \tfrac{1}{2}\sum_b\sqrt{Q_b(1-Q_b)}$ evaluated with each observed rate raised by $2\times10^{-3}$, the untrusted counterpart of Eq.~\eqref{eq:composed}. The trusted per-basis rows are evaluated from Eq.~\eqref{eq:composed} with $Q_Z \leqslant 0.13$ and $Q_X \leqslant 0.13$ imposed separately. At $\varepsilon = 0.10$ this gives $D_Z \leqslant 0.0375$ and $D_X \leqslant 0.13$ on the one-basis channel, hence $\FAE = 0.7631$, against $0.7846$ on the symmetric channel, where the two constraints coincide. Because the untrusted attacker is credited with the channel's noise as her own, absolute values are compared only between rows sharing a noise convention.

\begin{table}[htbp]
\caption{Amplitude damping, a non-Pauli channel outside the Pauli-channel cloning solutions of Refs.~\cite{cerf2000cloning, kerstan2026pauli}, and attacker ancilla size at the one-basis operating point ($\varepsilon = 0.10$). Both at threshold $\FAB \geqslant 0.87$.}
\label{tab:ad}
\begin{ruledtabular}
\begin{tabular}{cccc@{\qquad\qquad}ccc}
$\gamma$ & Honest $\FAB$ & $\FAE$ & Gates & Ancilla & $\FAE$ & Gates \\
\colrule
0.05 & 0.9812 & 0.8197 & 5 & 1 qubit  & 0.7850 & 6 \\
0.10 & 0.9622 & 0.7997 & 4 & 2 qubits & 0.7858 & 8 \\
0.15 & 0.9430 & 0.7751 & 4 & 3 qubits & 0.7858 & 5 \\
0.20 & 0.9236 & 0.7437 & 5 &          &        &   \\
\end{tabular}
\end{ruledtabular}
\end{table}

\section{Key-rate upper bounds and the existing literature}
\label{sm:keyrates}

\subsection{Constructed attacks as upper bounds}

Let $\rho_{ABE}$ be the state after the eavesdropper's interaction, with Alice's sifted bit $A$ uniform. Against collective attacks, the asymptotic one-way rate is the Devetak--Winter quantity~\cite{devetak2005distillation}
\begin{equation}
r = \inf_{\mathcal{E} \in \mathcal{C}} \left[ H(A|E)_{\mathcal{E}} - H(A|B)_{\mathcal{E}} \right],
\label{eq:dw}
\end{equation}
in which the infimum is taken over the set $\mathcal{C}$ of eavesdropping strategies compatible with the statistics Alice and Bob observe. A proof establishes a lower bound on that infimum by bounding $H(A|E)$ from below for every admissible $\mathcal{E}$. A constructed attack proceeds in the opposite direction, since exhibiting one particular $\mathcal{E}^\star \in \mathcal{C}$ evaluates the bound at a single point, and an infimum never exceeds any of its arguments, so that
\begin{equation}
r \;\leqslant\; H(A|E)_{\mathcal{E}^\star} - H(A|B)_{\mathcal{E}^\star}
\;=\; 1 - h(Q) - \Holevo(A{:}E)_{\mathcal{E}^\star}.
\label{eq:upper}
\end{equation}
The last equality uses $H(A|E) = H(A) - \Holevo(A{:}E) = 1 - \Holevo$ for a uniform sifted bit and $H(A|B) = h(Q)$ for a binary symmetric reconciliation channel at error rate $Q$. Every attack exhibited therefore certifies an upper bound on the key rate, and a stronger attack yields a lower upper bound. No feasibility question arises, because $\mathcal{E}^\star$ is compatible with the observed statistics by construction, the feasibility constraint of the search enforcing that compatibility.

The Holevo quantity is computed directly from the simulator. For each basis the attacker's two conditional memory states $\rho_E^{(0)}$ and $\rho_E^{(1)}$ are obtained as marginals of the post-interaction pure state, and
\begin{equation}
\Holevo = S\!\left(\tfrac{1}{2}\rho_E^{(0)} + \tfrac{1}{2}\rho_E^{(1)}\right)
- \tfrac{1}{2}S\!\left(\rho_E^{(0)}\right) - \tfrac{1}{2}S\!\left(\rho_E^{(1)}\right),
\label{eq:holevo}
\end{equation}
averaged over the two sifted bases, with $S(\cdot)$ the von Neumann entropy. Running the outer search against $\Holevo$ in place of $\FAE$, behind the same feasibility constraint, returns the values of Table~\ref{tab:keyrate}, with $\Holevo = h(Q)$ at every $Q$ tested, so that Eq.~\eqref{eq:upper} collapses onto $r \leqslant 1 - 2h(Q)$ and the rate is driven to zero at $Q = 0.11$. Neither the closed form nor the $11\%$ threshold was supplied to the search.

Constructive upper bounds of this kind are established practice. Explicit attacks saturate the device-independent collective-attack rate~\cite{acin2007diqkd, pironio2009diqkd, primaatmaja2023review}, attain device-dependent rates for families of imperfect sources~\cite{woodhead2014tight}, and bound device-independent rates through convex combinations of local and nonlocal behaviours~\cite{lukanowski2023convex}, and entropic upper bounds on device-dependent rates are available in closed form~\cite{kraus2005lower, renner2005information, moroder2006upper}. The present construction differs from those works in obtaining the attack by a search over circuits that takes no analytical input, which is what allows the same apparatus to operate on channels and drift processes for which no closed form exists.

\subsection{Comparison with semidefinite relaxations}
\label{sm:npa}

To render the comparison with certified methods quantitative, a small Navascu\'es--Pironio--Ac\'in relaxation for the CHSH scenario was implemented directly~\cite{navascues2008npa}. Alice and Bob hold dichotomic observables $A_0, A_1, B_0, B_1$ with $A_i^2 = B_j^2 = \mathbb{I}$ and $[A_i, B_j] = 0$. A monomial is a pair of words over the two alphabets, reduced by cancelling adjacent repeated letters, and the moment matrix carries entries $\Gamma_{uv} = \langle u^\dagger v \rangle$ indexed by a monomial set. Moments are taken real, the usual relaxation for CHSH. Requiring $\Gamma \succeq 0$ with $\Gamma_{11} = 1$ and maximising $\langle A_0B_0\rangle + \langle A_0B_1\rangle + \langle A_1B_0\rangle - \langle A_1B_1\rangle$ gives the results of Table~\ref{tab:npa}.

\begin{table}[htbp]
\caption{The relaxation reproduces Tsirelson's bound $2\sqrt{2}$~\cite{cirelson1980quantum} at three monomial sets, which validates the implementation. The dimension is the side of the moment matrix, and the unknowns are the number of distinct moments. Solved with SCS~\cite{odonoghue2016conic} at $\varepsilon = 10^{-9}$ on one core.}
\label{tab:npa}
\begin{ruledtabular}
\begin{tabular}{lcccc}
Level & Dimension & Unknowns & $S \leqslant$ & Solve time \\
\colrule
$1$      & 5  & 13 & 2.828427 & 0.02 s \\
$1{+}AB$ & 9  & 25 & 2.828427 & 0.03 s \\
$2$      & 13 & 41 & 2.828427 & 0.05 s \\
\end{tabular}
\end{ruledtabular}
\end{table}

The quantity of relevance here is the device-independent limit on what the eavesdropper learns, and the Tsirelson bound serves only as a check. Her guessing probability for Alice's outcome on input $x = 0$ is computed at a fixed observed CHSH value in the standard formulation, in which the observed behaviour is decomposed into two subnormalised behaviours, one per guess, each required to carry its own positive semidefinite moment matrix, with the two summing to a behaviour of CHSH value $S$. Maximising the success probability over that decomposition returns the values of Table~\ref{tab:pguess}, which agree with the analytical $P_{\mathrm{guess}}(S) = \tfrac{1}{2} + \tfrac{1}{2}\sqrt{2 - S^2/4}$ to one part in $10^{10}$.

\begin{table}[htbp]
\caption{Device-independent guessing probability at level $1{+}AB$, with two moment-matrix blocks of dimension nine. Agreement with the closed form holds at the solver tolerance, which validates the implementation.}
\label{tab:pguess}
\begin{ruledtabular}
\begin{tabular}{ccccc}
$S$ & $P_{\mathrm{guess}} \leqslant$ & Analytical & Residual & Solve time \\
\colrule
2.828 & 0.51229 & 0.51229 & $6.9 \times 10^{-10}$ & 0.25 s \\
2.750 & 0.66536 & 0.66536 & $1.9 \times 10^{-12}$ & 0.05 s \\
2.600 & 0.77839 & 0.77839 & $7.3 \times 10^{-12}$ & 0.08 s \\
2.444 & 0.85592 & 0.85592 & $8.8 \times 10^{-11}$ & 0.04 s \\
2.200 & 0.94441 & 0.94441 & $5.0 \times 10^{-11}$ & 0.05 s \\
\end{tabular}
\end{ruledtabular}
\end{table}

\subsection{upper bounds and lower bounds}
\label{sec:upper bounds_lower bounds}

Semidefinite relaxations of the Navascu\'es--Pironio--Ac\'in hierarchy bound the eavesdropper's information from above at any Hilbert-space dimension~\cite{navascues2008npa, brown2021npa}, and they return an upper bound without an attack. Converting the guessing probability into the Holevo quantity uses $H(A|E) \geqslant H_{\min}(A|E) = -\log_2 P_{\mathrm{guess}}$, so that $\Holevo = 1 - H(A|E) \leqslant 1 + \log_2 P_{\mathrm{guess}}$. At the operational threshold $S = \tau_S$ the guessing-probability relaxation gives $\Holevo \leqslant 0.776$ and the Ac\'in envelope, restricted to collective attacks, gives $\Holevo \leqslant 0.607$, whereas the attacks constructed in Sec.~\ref{sec:libraries} attain $\Holevo \approx 0.350$ averaged over the operating noise (Table~\ref{tab:bound}). The min-entropy upper bound is looser than the envelope at every $S$, as it must be, since the min-entropy bounds the conditional von Neumann entropy from below whereas the envelope is the tight bound on the von Neumann quantity itself. A tighter device-independent upper bound would follow from the conditional-entropy relaxation of Ref.~\cite{brown2021npa}, which solves a sequence of programmes at one quadrature node each and was not evaluated here. The recomputed envelope at $S = \tau_S$ returns $\Holevo \leqslant 0.6069$, which reproduces the value quoted in the main text from an independent implementation. The interval between the upper bounds and the lower bound arises for the two reasons given in Sec.~\ref{sec:libraries}, the lower bound being averaged over the operating noise and restricted to a three-qubit preparation, and it does not measure looseness in either upper bound.

\begin{table}[htbp]
\caption{Upper and lower bounds on the eavesdropper's Holevo information. The first two columns hold without reference to Hilbert-space dimension, the first through the min-entropy and the second as the tight collective-attack bound on the von Neumann quantity, and the third is exhibited by a circuit. The attained value is averaged over the operating noise distribution, which is why it sits furthest below.}
\label{tab:bound}
\begin{ruledtabular}
\begin{tabular}{cccc}
$S$ & NPA via min-entropy & Ac\'in envelope & Attained by construction \\
\colrule
2.828 & 0.0350 & 0.0021 & --- \\
2.600 & 0.6386 & 0.4184 & --- \\
2.444 & 0.7755 & 0.6069 & $\approx 0.350$ \\
\end{tabular}
\end{ruledtabular}
\end{table}

\subsection{Distribution of computational cost}
\label{sm:cost}

The main text records that a relaxation is the cheaper instrument for the static device-independent quantity, and that comparison rests on direct measurement. Enlarging the Bell scenario in the direction that enlarges the moment matrix, namely the number of measurement settings $m$ per party, gives Table~\ref{tab:scaling}. Solve time grows as $m^{4.7}$ and the number of distinct moments as $m^{4.3}$ at level two, so the relaxation does become expensive, and at eight settings it reaches $42.6$~s, which is parity with a single operating point of the circuit search.

\begin{table}[htbp]
\caption{Growth of the relaxation with the number of measurement settings per party, at two monomial sets. Dimension is the side of the moment matrix. The circuit search costs $43$--$49$~s per operating point independently of $m$, since it optimises over circuits and not over correlations.}
\label{tab:scaling}
\begin{ruledtabular}
\begin{tabular}{ccccccc}
 & \multicolumn{3}{c}{level $1{+}AB$} & \multicolumn{3}{c}{level $2$} \\
$m$ & dim & moments & time & dim & moments & time \\
\colrule
2 & 9  & 25   & 0.10 s & 13  & 41    & 0.06 s \\
3 & 16 & 100  & 0.14 s & 28  & 244   & 0.24 s \\
4 & 25 & 289  & 0.33 s & 49  & 865   & 0.80 s \\
5 & 36 & 676  & 0.59 s & 76  & 2276  & 2.17 s \\
6 & 49 & 1369 & 1.20 s & 109 & 4969  & 8.66 s \\
8 & 81 & 4225 & 4.78 s & 193 & 16769 & 42.64 s \\
\end{tabular}
\end{ruledtabular}
\end{table}

These figures support parity and nothing stronger. In the CHSH scenario studied here the relaxation is three orders of magnitude cheaper than the search, and a crossover obtained by growing the scenario until the convex problem becomes hard is not an argument for the search. Two observations remain after that qualification. The scaling directions differ, since the relaxation grows with the complexity of the correlation whereas the search grows with the ancilla dimension and the circuit depth, so the two are cheap in different places. In addition, the search returns an explicit attack, an object the relaxation does not produce, and the sequential problem of Sec.~\ref{sec:drift_formulation} requires such objects as its actions.

\subsection{Comparison with device-independent upper bounds}

Two families of upper bounds are established for CHSH-based device-independent key rates. Arnon-Friedman and Leditzky~\cite{arnonfriedman2021upper} bound the rate in terms of the observed violation alone, using entropic and entanglement-measure arguments that do not refer to any particular attack, which yields a bound holding uniformly over all devices producing that violation. {\L}ukanowski, Balanz\'o-Juand\'o, Farkas, Ac\'in and Ko{\l}ody\'nski~\cite{lukanowski2023convex} instead construct an explicit adversary, decomposing the observed correlation as a convex mixture of a local part, which the eavesdropper prepares and therefore knows completely, and a nonlocal remainder, so that the key rate is bounded by the weight of the local component. The related result of Ref.~\cite{farkas2021bell} uses the same machinery to show that Bell nonlocality alone does not suffice for security of the standard protocols.

Both are upper bounds, and the present construction supplies a lower bound beneath them, being closest in approach to the second family while differing in the quantity optimised. A convex-combination attack fixes the \emph{form} of the adversary, a local-deterministic mixture, and optimises the decomposition weight, whereas the search reported here fixes nothing about the form and optimises over circuits realising the strategy. The two are complementary in an operational sense. A convex-combination bound applies to any device reproducing the correlation and is therefore device-independent in the full sense, whereas the present bound applies to an adversary restricted to the stated ancilla dimension and noise model, so it is the weaker statement, although a physically specified attacker attains it. Where the two overlap, the convex-combination bound should be the tighter, and confirming that on shared CHSH values is the natural first check.

Since all three quantities are functions of the CHSH value $S$, the comparison requires no protocol translation. Each library cell yields an upper bound through Eqs.~\eqref{eq:holevo} and~\eqref{eq:upper}, the bound of Ref.~\cite{arnonfriedman2021upper} overlays as a closed-form function of $S$, and the convex-combination bound of Ref.~\cite{lukanowski2023convex} follows from the local weight returned by the linear programme of that work. A crossing at intermediate $S$ would locate correlations for which an explicit quantum strategy of bounded memory improves on the local-mixture adversary. The comparison has not been carried out here, and no claim is made regarding its outcome.

One qualification applies. The attacks reported here are constrained to a fixed ancilla dimension, so the upper bound they certify is not device-independent, and a device outside that specification could leak more. The comparison is therefore diagnostic of how much of the device-independent bound is reachable by a quantum adversary with bounded memory, and a large residual gap at a given $S$ would indicate that attaining the device-independent bound at that point requires a larger adversarial memory than the one granted here, not that the bound is loose.

\section{Benchmarking against numerical key-rate methods}
\label{sm:wlc}

\subsection{Comparison parameters}

The method of Winick, L\"utkenhaus and Coles~\cite{winick2018reliable} is the established route to reliable key rates for device-dependent protocols whose statistics admit no closed form. It bounds the asymptotic collective-attack rate from below by
\begin{equation}
r \;\geqslant\; \min_{\rho_{AB} \in \mathcal{S}}\;
D\!\left(\rho_{AB} \,\|\, \mathcal{Z}_A(\rho_{AB})\right) \;-\; \mathrm{leak_{EC}},
\label{eq:wlc}
\end{equation}
where $\mathcal{S}$ is the set of states compatible with the observed statistics, $\mathcal{Z}_A$ pinches Alice's qubit in the key basis, and $\mathrm{leak_{EC}} = h(Q_Z)$ for one-way reconciliation at the Shannon limit. A practical reconciliation efficiency $f > 1$ scales that term identically on both sides of the comparison, so it shifts both bounds by the same amount and leaves the gap unchanged. The present work bounds the same rate from above by exhibiting an attack. The two become directly comparable once both are given the same input, taken here to be the per-basis error rates $(Q_Z, Q_X)$ with the eavesdropper held responsible for the whole disturbance, which is the standard conservative convention.

\subsection{Evaluation of the lower bound}

Pauli twirling maps any two-qubit state to a Bell-diagonal one without altering the per-basis error rates or raising the key rate, so the minimisation in Eq.~\eqref{eq:wlc} may be restricted to Bell-diagonal states. Writing the Bell weights as $(a,b,c,d)$ for $(\Phi^+, \Phi^-, \Psi^+, \Psi^-)$ gives $Q_Z = c+d$ and $Q_X = b+d$, and pinching Alice's qubit removes the coherences between the two blocks the key basis selects, leaving
\begin{equation}
D\!\left(\rho \,\|\, \mathcal{Z}_A(\rho)\right) = 1 + h(Q_Z) - H(a,b,c,d)
\label{eq:twirled}
\end{equation}
in bits, with $H$ the Shannon entropy of the Bell weights. Subtracting the leakage leaves $r \geqslant 1 - \max H(a,b,c,d)$, a one-dimensional maximisation in the free parameter $d$ carried out on a grid. The optimum sits at $d = Q_Z Q_X$, so the procedure returns $1 - h(Q_Z) - h(Q_X)$, and that identity serves as a check on the numerics without being assumed by them. Each evaluation takes about one second.

\subsection{Evaluation of the upper bound}

The attack search of Sec.~\ref{sm:framework} is run with the feasibility constraint placed on both per-basis rates simultaneously, so that a circuit counts as feasible only when it reproduces $(Q_Z, Q_X)$ to within $2 \times 10^{-3}$, and with the objective set to the Holevo information the eavesdropper holds about the \emph{key-basis} bit alone. Resolving the Holevo quantity per basis is essential here. An attack may hold considerable information about the complementary basis while learning nothing about the key, so averaging the two would credit the eavesdropper with information that never enters the rate and would yield an upper bound below the proved lower bound. The upper bound is then evaluated at the rates the circuit produces, which may differ from the targets within the tolerance, so that what is certified is a statement about the statistics the circuit presents.

\begin{table}[htbp]
\caption{Constructed attacks against the reliable numerical lower bound of Ref.~\cite{winick2018reliable} at matched per-basis error rates. Attacks use a two-qubit ancilla and are matched to the per-basis rates shown. $\Holevo(A_Z{:}E)$ is the information the circuit holds about the key-basis bit and $h(Q_X)$ the largest value any attack could hold at that phase error rate. The final columns bound the asymptotic key rate from above, as $1 - h(Q_Z) - \Holevo(A_Z{:}E)$ for the circuit found, and from below by the numerical minimisation.}
\label{tab:wlc-full}
\begin{ruledtabular}
\begin{tabular}{cccccccc}
$Q_Z$ & $Q_X$ & $\Holevo(A_Z{:}E)$ & $h(Q_X)$ & Gates & $r \leqslant$ & $r \geqslant$ & Gap \\
\colrule
0.0501 & 0.0519 & 0.2941 & 0.2944 & 11 & 0.4191 & 0.4188 & $3\times10^{-4}$ \\
0.0684 & 0.0320 & 0.2043 & 0.2043 & 5  & 0.4356 & 0.4356 & $<10^{-4}$ \\
0.0895 & 0.0120 & 0.0938 & 0.0938 & 4  & 0.4716 & 0.4716 & $<10^{-4}$ \\
0.0992 & 0.0020 & 0.0202 & 0.0204 & 6  & 0.5135 & 0.5132 & $3\times10^{-4}$ \\
\end{tabular}
\end{ruledtabular}
\end{table}

\subsection{Convergence diagnostics}

At the two intermediate points, a first pass run at $400$ evaluations per candidate and $18$ generations returned $\Holevo(A_Z{:}E) = 0.1308$ against $h(Q_X) = 0.2040$ and $0.0573$ against $0.0936$, leaving upper bounds above the lower bound by $0.073$ and $0.036$ in rate. Because the lower bound is a proven lower bound, the discrepancy could lie only in the search. Raising the evaluation budget from $400$ to $900$ per candidate and the outer loop from $18$ to $24$ generations closed both gaps to below $10^{-4}$. The converged runs also returned smaller circuits, of five and four gates against the earlier six and eight, consistent with the budget argument of Sec.~\ref{sm:framework}.

The comparison supplies the practical argument for computing both bounds. A constructive upper bound on its own cannot distinguish an attack that saturates the truth from one that a limited budget has failed to improve. Placed against a proved lower bound, the same number acquires a convergence criterion, since where the gap closes the search has converged and where it remains open the search requires more budget. Neither method supplies that diagnostic alone, and the diagnostic allows a heuristic search to be reported as a bound.

\subsection{Relation to noisy-channel key-rate studies}

A body of recent work computes key rates for specific noisy channels in place of the abstract statistics considered here. Kish \textit{et al.} compare discrete- and continuous-variable protocols under phase noise in the thermal-loss channel~\cite{kish2024comparison}, Mao \textit{et al.} treat passive continuous-variable operation over a turbulent atmospheric link~\cite{mao2022passive}, and Shu \textit{et al.} obtain rates over noisy channels by a testing-state construction~\cite{shu2023testing}. All three compute achievable rates, so all three lie on the same side of the gap as Ref.~\cite{winick2018reliable} and are complementary to the upper bounds reported here. Two obstacles prevent a like-for-like numerical comparison at present. The continuous-variable results are stated for an infinite-dimensional carrier under a Gaussian or near-Gaussian modulation, whereas the search operates on a finite-dimensional ancilla, so the two are not evaluated on a common state space. In addition, the channel models differ, since thermal loss and atmospheric turbulence are not Pauli maps, and although the search extends to non-Pauli noise, as Table~\ref{tab:ad} shows for amplitude damping, it has not been run on either channel. Extending the construction to the thermal-loss channel and placing an upper bound beside the rates of Ref.~\cite{kish2024comparison} at matched excess noise and transmissivity is the natural next comparison, and requires no change to the apparatus beyond the channel model.

\section{Translation of discovered circuits to eavesdropping strategies}
\label{sm:translation}

No property of a discovered attack depends on its expression as a gate sequence. This section provides the translation into the conventional language of channels and measurements, so that the output of the search may be used by methods that take a strategy as input.

\subsection{Circuit representation as an isometry}

The search returns a unitary $U_E(T, \theta)$ acting on the travelling qubit together with an ancilla initialised in $\ket{0}^{\otimes n_a}$. Fixing the ancilla input defines an isometry
\begin{equation}
V: \mathcal{H}_A \to \mathcal{H}_B \otimes \mathcal{H}_E,
\qquad
V\ket{\psi} = U_E\!\left(\ket{\psi} \otimes \ket{0}^{\otimes n_a}\right),
\label{eq:isometry}
\end{equation}
with $V^\dagger V = \mathbb{I}$. The map $V$ is the Stinespring dilation~\cite{stinespring1955positive} of the channel Bob sees, and the eavesdropper's strategy is fully specified by $V$, the gate decomposition being one factorisation among many and carrying no physical content.

\subsection{Kraus operators and induced channels}

Choosing a basis $\{\ket{e_k}\}_{k=1}^{d_E}$ for $\mathcal{H}_E$, with $d_E = 2^{n_a}$, gives Kraus operators
\begin{equation}
K_k = \left(\mathbb{I}_B \otimes \bra{e_k}\right) V ,
\qquad \sum_k K_k^\dagger K_k = \mathbb{I},
\label{eq:kraus}
\end{equation}
each a $2 \times 2$ matrix. Bob's channel is $\mathcal{N}(\rho) = \sum_k K_k \rho K_k^\dagger$ and the complementary channel carrying the attacker's side information is $\mathcal{N}^c(\rho) = \sum_{k,l} \Tr\!\left(K_l^\dagger K_k \rho\right) \ket{e_k}\!\bra{e_l}$. For a two-qubit ancilla this gives four Kraus operators of size $2 \times 2$, a description of at most $32$ real parameters that replaces the gate list entirely. The composition with the physical noise process is $\Lambda \circ \mathcal{N}$ on Bob's side, while the attacker's marginal is $\mathcal{N}^c$ alone, which is the formal statement of the decoupling used throughout Secs.~\ref{sec:setting} and~\ref{sm:sweeps}. The BB84 drift library of Sec.~\ref{sec:libraries} places the channel before the interaction instead, $\mathcal{N} \circ \Lambda$, and there the attacker's marginal depends on the noise. Because the algorithmic architecture treats the ancilla size $n_a$ as an initialisation parameter, the framework accommodates a larger adversarial dimension $d_E$ without structural modification of the solver.

\subsection{Measurements and resulting correlations}

The strategy is completed by the measurement the attacker performs after the public announcement. For a prepare-and-measure protocol the announcement of basis $b$ selects a two-outcome measurement $\{M_0^{(b)}, M_1^{(b)}\}$ on $\mathcal{H}_E$, and the optimal choice is the Helstrom measurement~\cite{helstrom1969quantum} for discriminating $\rho_E^{(0)} = \mathcal{N}^c(\ket{\psi_0^b}\!\bra{\psi_0^b})$ from $\rho_E^{(1)}$, namely the projector onto the positive part of $\rho_E^{(0)} - \rho_E^{(1)}$. Her success probability is then given by Eq.~\eqref{eq:helstrom}, which is how the main text's figure of merit arises without reference to any circuit.

For the entanglement-based case, the same isometry acts on one half of the shared pair, the correlation is $p(ab|xy) = \Tr\!\left[\left(A_{a|x} \otimes B_{b|y}\right)\rho_{AB}\right]$ with $\rho_{AB}$ the post-interaction reduced state, and the attacker retains $\rho_E^{ab|xy}$ as classical-quantum side information. A CHSH value and a Holevo quantity follow from these by the Born rule alone. The object handed to a relaxation is thus the triple $(\rho_{AB}, \{A_{a|x}\}, \{B_{b|y}\})$ together with the Kraus set of Eq.~\eqref{eq:kraus}, all in the standard form.

\subsection{Application as an ansatz}
\label{sec:npa_integration}

A discovered circuit, translated in this way, is a feasible point for the conventional methods, and three applications follow from the translation. Firstly, a see-saw optimisation over states and measurements requires an initial point, and a discovered circuit supplies one that already meets the constraint, which is material because see-saw optimisation is sensitive to initialisation. Secondly, a lower-bounding computation of the conditional entropy requires a candidate attack against which to check its output, and any separation between the certified bound and the value of an explicit attack localises the region in which the relaxation is loose. Thirdly, a feasibility question of the form ``is this correlation attainable at this noise level with an ancilla of this size'' is answered affirmatively by exhibiting the circuit, the isometry constituting an explicit witness. In each case the search contributes structure and the relaxation contributes certification.

\section{Drift environments, reference policies and reproducibility}
\label{sm:drift}

\subsection{Environment specification}

Both environments are episodic with $T = 50$ rounds. The noise level is initialised from $\mathcal{N}(\mu, \sigma^2)$ and advanced by the exact one-round update of Eq.~\eqref{eq:ou}, $p_{t+1} = \mu + (p_t - \mu)e^{-\kappa} + \sigma\sqrt{(1 - e^{-2\kappa})/2\kappa}\,\xi_t$ with $\xi_t \sim \mathcal{N}(0,1)$, and clipped to the ranges of Table~\ref{tab:mapping}. For each library circuit the single-round statistic and reward are tabulated on a grid in $p$ and interpolated linearly. Detection is declared at the end of an episode if the mean of the single-round statistic over its $50$ rounds lies below the threshold. The observation comprises the current $p_t$, the three preceding values, the previous action, the mean statistic over the last ten rounds, the mean reward so far and the fraction of rounds remaining. The feasibility mask admits the circuits whose statistic, appended to the preceding nine rounds, keeps the ten-round mean above the threshold, with a slack of $10^{-3}$ on E91; the greedy oracle picks the admissible circuit of largest reward under the same rule.

In the BB84 environment the bit-flip channel acts on Alice's state before the attacker's two-qubit unitary on the travelling qubit and one ancilla, following Ref.~\cite{decker2024qkd}; $\FAB$ is the mean over the four sifted states of Bob's matched-basis recovery probability, which gives $1 - p/2$ for the passive attacker, and $\FAE$ is the corresponding probability for the ancilla measured in the announced basis. The noise-indexed conservative lookup plays the most conservative of the four library levels at the anchor nearest $p_t$, and so selects among six circuits. In the E91 environment the attacker prepares a three-qubit state on $A$, $B$ and $E$, the channel $\Lambda_p^A \otimes \Lambda_p^B$ acts on it, $S$ is evaluated at the settings stored with each cell, and $\Holevo$ is computed for Alice's $Z$ outcome.

\subsection{Checkpoint selection and spread}

Agent values in Table~\ref{tab:results} are those of the best checkpoint on evaluation. On BB84 the three best strict checkpoints reach $0.7434$, $0.7433$ and $0.7432$, all without detection in $50$ episodes, and the three best tolerated checkpoints $0.7477$, $0.7463$ and $0.7462$ at $4\%$, $14\%$ and $4\%$ detection. A retraining of the strict agent for $5 \times 10^{5}$ steps in a refactored copy of the environment reached $0.7348$ at $2\%$ detection, and a policy-gradient agent $0.6687$ without detection, so the match between the agent and the greedy oracle on BB84 is a best-case outcome of training, whereas the ordering of Table~\ref{tab:results} is robust to it.

\subsection{Values computed for this work}

The E91 rows marked $^\dagger$ in Table~\ref{tab:results} were computed from the released E91 library with an independent reimplementation of the environment, which reproduces the stored $(S, \Holevo)$ value of every library cell to $10^{-6}$. On $2000$ episodes (seeds $1000$--$2999$), each of the $27$ actions was played in every round, and the best circuit without detection in any episode is the ``safe'' cell anchored at $p = 0.0642$, with $\Holevo = 0.135$; the next is the ``safer'' cell anchored at $p = 0.0428$, with $0.125$. The greedy oracle reaches $0.345$ on the same episodes. The zero-detection upper bound was obtained by backward induction over $(t, p_t, E_t)$ on $100$ noise levels with a Gaussian transition matrix and a grid of spacing $0.002$ in the accumulated excess $E_t$, with terminal value $-\infty$ for $E_T < 0$. Rounding $E_{t+1}$ up to the grid, which favours the attacker, gives $0.3575$, and rounding it down gives $0.3550$, so $V_{\mathrm{DP}} = 0.356 \pm 0.001$. The dynamic-programming run of the released experiments discretised the accumulated statistic far more coarsely, which is why it both falls below this upper bound and incurs detection. For BB84 the released dynamic programme uses $100$ noise levels and $2000$ bins of accumulated $\FAB$ and incurred no detection in $200$ evaluation episodes, and its value $0.7527$ is reported unchanged. The best genuinely fixed BB84 circuit could not be evaluated, because the library files of that environment are not in the released artefact; the accompanying script \texttt{best\_fixed\_and\_dp.py} computes it, together with the E91 values above, once they are supplied.

\section{Numerical cost and reproducibility}
\label{sm:compute}

All results were produced on commodity CPU hardware without accelerators, the inner-loop arithmetic being dominated by density-matrix simulation of two to four qubits. The BB84 library required roughly $20$ CPU-hours across $24$ cells, the E91 library $60$--$100$ CPU-hours across $26$ cells, and adaptive-attack training $30$--$50$ CPU-hours per protocol across the full policy comparison, for a total near $150$--$250$ CPU-hours. The channel sweeps of Sec.~\ref{sm:sweeps}, including the per-basis repetition of Sec.~\ref{sm:perbasis}, add about $0.6$ CPU-hours, and the semidefinite relaxations of Sec.~\ref{sm:npa} add a few seconds in total. Peak memory stayed below $2$~GB.

The released artefact contains the search framework, the BB84 and E91/DIQKD simulation harnesses, the sweep scripts, the key-rate and semidefinite-relaxation scripts, the training and evaluation scripts, the analysis notebooks, the random seeds, an environment specification and a top-level script regenerating every reported value from stored results.

\section{Scope and dual-use considerations}
\label{sm:scope}

The circuits reported here are attack strategies on QKD protocols, published openly in accordance with the disclosure norms of academic cryptography. Three considerations limit the realistic risk. Firstly, the experiments run at the analytical reference scale of the public literature, with two-qubit ancillas, idealised Pauli, amplitude-damping and bilateral depolarising channels, and targets that are published bounds~\cite{cerf2000cloning, fuchs1997optimal, bruss1998optimal, acin2007diqkd}, and no result exceeds those bounds, so the information-theoretic upper bound constraining a deployed attacker is untouched. Secondly, the distance between a simulated attack in idealised noise and an attack on a fielded link is dominated by physical-layer factors not modelled here, among them detector control, finite-key statistics, decoy-state calibration, source imperfections and authenticated post-processing, which is historically where operationally usable attacks have originated. Thirdly, adaptive eavesdropping under channel drift is the threat model that informs how users choose abort thresholds and recalibration intervals, so characterising what such an attacker extracts serves the defender. The release contains no hardware drivers, side-channel primitives or post-processing attack code.

\section{Note on concurrent work}
\label{sm:concurrent}

A preliminary version of this work was presented at QCrypt 2026. During the preparation of this manuscript, we also became aware of independent, concurrent work by Bagunu (2026) \cite{bagunu2026adaptive}. While both studies model adaptive eavesdropping in quantum key distribution as a sequential decision process, their scopes and methodologies differ significantly. Bagunu employs tabular reinforcement learning to modulate a fixed intercept-resend attack against a stationary channel. In contrast, our approach addresses a more advanced threat model: our deep reinforcement learning agent dynamically selects from a library of variationally constructed quantum circuits (where gate topology and rotation angles are jointly optimised) to exploit a continuously drifting channel.

\end{document}